\documentclass[11pt]{article}

\usepackage[T1]{fontenc}
\usepackage{amsfonts}
\usepackage{alltt}
\usepackage{amssymb,verbatim,ifthen}
\usepackage{dsfont}
\usepackage[all]{xy}
\usepackage{color}
\usepackage{xcolor}
\usepackage{graphicx}

\usepackage{amsthm}
\usepackage{mathtools}
\usepackage{mathabx}
\usepackage{pdfsync}
\usepackage{listliketab}
\usepackage[authoryear,round]{natbib}
\usepackage{subcaption} 
\usepackage[hidelinks]{hyperref}
\usepackage{textcomp}

\usepackage{xparse}

\usepackage{booktabs} 
\usepackage[ruled]{algorithm2e} 

\SetAlFnt{\small}
\SetAlCapFnt{\small}
\SetAlCapNameFnt{\small}
\SetAlCapHSkip{0pt}
\IncMargin{-\parindent}

\numberwithin{equation}{section}  
\usepackage[sort&compress,capitalize,nameinlink]{cleveref}

\theoremstyle{plain}
\newtheorem{theorem}{Theorem}
\newtheorem{claim}[theorem]{Claim}
\newtheorem{corollary}[theorem]{Corollary}
\newtheorem*{corollary*}{Corollary}
\newtheorem{lemma}[theorem]{Lemma}

\theoremstyle{definition}
\newtheorem{definition}[theorem]{Definition}
\newtheorem*{definition*}{Definition}

\newtheorem{assumption}[theorem]{Assumption}

\newtheorem*{hypothesis*}{Hypothesis}

\newtheorem{observation*}[theorem]{Observation}

\theoremstyle{remark}
\newtheorem{remark}[theorem]{Remark}
\newtheorem*{remark*}{Remark}
\newtheorem*{notation*}{Notational remark}

\newif\ifopenquestions
\openquestionsfalse
\newif\ifmultischeme
\multischemefalse

\usepackage{iftex}

\usepackage[acronym]{glossaries}

\makeglossaries

\usepackage{acronym} 
\usepackage[textwidth=10mm]{todonotes}
\usepackage{booktabs}       
\usepackage{enumitem}
\usepackage{latexsym}
\usepackage{wasysym}
\usepackage{xspace}
\usepackage{fullpage}
\usepackage[onehalfspacing]{setspace}
\usepackage{longtable}

\definecolor[named]{mygreen}{cmyk}{0.97,0,0.53,0.29}
\newcommand{\debug}[1]{#1}

\newcommand{\braced}[1]{\ifthenelse{\equal{#1}{}}{}{\left(#1 \right)}}
\newcommand{\sbraced}[1]{\ifthenelse{\equal{#1}{}}{}{(#1)}}
\NewDocumentCommand{\E}{O{p}  }{\mathbb E\left[ #1 \right]}
\NewDocumentCommand{\N}{ }{\mathbb  N}
\NewDocumentCommand{\proba}{O{p} O{} }{\mathbb P \left[ #1 \ifthenelse{\equal{#2}{}}{}{\mid #2} \right]}

\NewDocumentCommand{\actions}{ }{\debug{A}}
\newglossaryentry{actions}{
    name={$\actions$},
    description={Set of actions of every player, that is $\{0, \ldots, \period-1\}$},
    sort={a}
}

\NewDocumentCommand{\pl}{ }{\debug{i}}

\NewDocumentCommand{\plb}{ }{\debug{j}}
\newglossaryentry{player}{
    name={$\pl$, $\plb$},
    description={Typical players in $\players$},
    sort={i}
}

\NewDocumentCommand{\nplayers}{ }{\debug{I}}

\NewDocumentCommand{\players}{ }{\debug{I}}
\newglossaryentry{players}{
    name={$\nplayers$},
    description={Set of players and its cardinality, when there is no ambiguity},
    sort={I}
}

\NewDocumentCommand{\socialcost}{O{\tmixedprof}  }{\debug{\mathsf{SC}(#1)}}
\newglossaryentry{socialcost}{
name={$\socialcost$},
description={Social cost of action profile $\mixedprof$, that is $\sum_{\pl \in \players} \cost$},
sort={SC}
}

\NewDocumentCommand{\rcost}{O{\ti}}{\debug{l_{#1}}}

\NewDocumentCommand{\strat}{O{\pl}}{\ifthenelse{\equal{#1}{}}{\debug{\boldsymbol x}}{\debug{x}^{#1}}}
\NewDocumentCommand{\tstrat}{O{\pl} O{\ti}}{\debug{x}^{#1}_{#2}}

\newglossaryentry{strat}{
    name={$\strat$},
    description={Strategy of player $\pl$ in a repeated context, \emph{\ie} a function $\histset \rightarrow \simplex(\actions)$},
    sort={xs}
    }
\newglossaryentry{tstrat}{
    name={$\tstrat$},description={Strategy of player $\pl$ at time $\ti$, \emph{\ie} $\strat(\hist)$},
    sort={xt}
}

\NewDocumentCommand{\corr}{O{}}{\sigma^{#1}}

\NewDocumentCommand{\s}{O{}  }{\sigma ^{#1}}

\NewDocumentCommand{\purep}{O{}}{\debug{\ifthenelse{\equal{#1}{\pl}}{a^{#1}}{
\ifthenelse{\equal{#1}{\plb}}{a^{#1}}{\boldsymbol a^{#1}}}}}
\NewDocumentCommand{\pureb}{O{}}{\debug{\ifthenelse{\equal{#1}{\pl}}{b^{#1}}{
\ifthenelse{\equal{#1}{\plb}}{b^{#1}}{\boldsymbol b^{#1}}}}}
\NewDocumentCommand{\act}{O{\pl}}{\purep[#1]}
\NewDocumentCommand{\actb}{O{\pl}}{\pureb[#1]}
\newglossaryentry{purep}{
    name={$\purep$, $\pureb$},
    description={Typical pure action profiles, with $\purep[\pl]$, $\pureb[\pl]$ the actions played by $\pl$},
    sort={ab}
}

\NewDocumentCommand{\mixedp}{O{}}{
    \debug{
        \ifthenelse{\equal{#1}{\pl}}
            {\sigma^{#1}}
            {\ifthenelse{\equal{#1}{\plb}}
                {\sigma^{#1}}
                {\boldsymbol \sigma^{#1}}}
    }
}
\newglossaryentry{mixedp}{
    name={$\mixedp$},
    description={Typical mixed action profiles, with $\mixedp[\pl]$ the mixed action played by $\pl$ and $\mixedp[\pl](\act)$ the probability for $\pl$ to play $\act$},
    sort={s}
}

\NewDocumentCommand{\tmixed}{O{\ti} O{\pl}}{
    \debug{
        \ifthenelse{\equal{#2}{\pl}}
            {\sigma_{#1}^{#2}}
            {\ifthenelse{\equal{#2}{\plb}}
                {\sigma_{#1}^{#2}}
            {\boldsymbol \sigma^{#2}_{#1}}    }
    }
}
\newglossaryentry{tmixed}{
    name={$\tmixed$},
    description={Typical mixed action profiles, with $\tmixed[\pl]$ the mixed action played by $\pl$ and $\tmixed[\pl](\act)$ the probability for $\pl$ to play $\act$},
    sort={s}
}
\NewDocumentCommand{\correq}{}{\debug{\boldsymbol{\tau}}}

\NewDocumentCommand{\tcorreq}{O{\ti}}{\correq_{#1}}

\NewDocumentCommand{\tact}{O{\ti} O{\pl}  }{\debug{\ifthenelse{\equal{#2}{\pl}}{a^{#2}_{#1}}{
\ifthenelse{\equal{#2}{\plb}}{a^{#2}_{#1}}{\boldsymbol a^{#2}_{#1}}}}}

\newglossaryentry{timedaction}{
    name={$\tact$},
    description={In a repeated setting, action played by $\pl$ at time $\ti$},
    sort={at}
}

\NewDocumentCommand{\vent}{}{\boldsymbol{k}}
\NewDocumentCommand{\tvent}{}{\boldsymbol{k}_{\ti}}
\NewDocumentCommand{\tcost}{O{\tmixedprof} O{\pl} O{\tvent} }{\debug{c^{#2}_{#3}\sbraced{#1}}}
\NewDocumentCommand{\ticost}{O{\mixedp} O{\pl} O{\tvent} }{\debug{c^{#2}_{#3}\sbraced{#1}}}
\NewDocumentCommand{\cost}{O{\mixedp} O{\pl} O{} }{\debug{c^{#2}_{#3}(#1)}}
\newglossaryentry{cost}{name={$\cost[\purep] \text{\normalfont\ or } \tcost[\purep]$}, description={Cost of player $\pl$ if players follow profile $\purep$ with $\vent$ left implicit when it is not ambiguous},
sort={cost}
}

\NewDocumentCommand{\states}{}{\debug{S}}

\newglossaryentry{states}{
    name={$\states$},
    description={State of the system, so the number of jobs that every player holds},
    sort={S}}

\NewDocumentCommand{\st}{O{\ti}}{\debug{\boldsymbol{k}_{#1}}}
\newglossaryentry{state}{
    name={$\st$},
    description={State of the system at period $\ti$},
    sort={kt}}

\NewDocumentCommand{\hist}{O{\ti}}{\debug{h_{#1}}}
\NewDocumentCommand{\histset}{}{\mathcal{\debug{H}}}
\newglossaryentry{history}{name={$\hist$},description={History: state-action profile pairs of previous periods}, sort={h}}

\NewDocumentCommand{\ti}{ }{\debug{t}}

\NewDocumentCommand{\numsim}{O{\act[\plb]} }{\debug{N}_{#1}}
\NewDocumentCommand{\numpre}{O{\act[\plb]} }{\debug{M}_{#1}}

\NewDocumentCommand{\tnent}{O{\ti}}{\debug{k_{#1}}}
\NewDocumentCommand{\tent}{O{\ti} O{\pl}}{\debug{k_{#1}^{#2}}}
\newglossaryentry{tjobs}{
    name={$\tnent, \tent$},
    description={Number of jobs of all players, of player $\pl$. See $\nent$ and $\ent$},
    sort={k}
}

\NewDocumentCommand{\latecount}{O{\pl}}{\debug{\tilde k}^{#1}}
\newglossaryentry{latecount}{
    name={$\latecount$},
    description={Number of jobs of player $\pl$ which arrive late},
    sort={ktilde}}

\NewDocumentCommand{\late}{O{\pl} m}{\debug{p^{#1}(#2)}}
\newglossaryentry{late}{
    name={$\late{\purep}$},
    description={Probability for a job of $\pl$ to be late assuming that the players action profile is $\purep$},
    sort={p}
}

\NewDocumentCommand{\ent}{O{\pl}  }{\debug{k^{#1}}}

\newglossaryentry{jobs}{
    name={$\ent$},
    description={Number of jobs of players $\pl$},
    sort={kp}
}

\NewDocumentCommand{\nent}{}{\debug{k}}
\newglossaryentry{totaljobs}{
    name={$\nent$},
    description={Total number of jobs},
    sort={ktotal}
}

\NewDocumentCommand{\plty}{ }{\debug{C}}
\NewDocumentCommand{\tplty}{ O{\tnent}}{\debug{C_{#1}}}

\newglossaryentry{penalty}{name={$\plty$, $\tplty$},description={Penalty incurred by a job that arrives after the end of a period},
sort={C}}

\NewDocumentCommand{\period}{ }{\debug{L}}
\newglossaryentry{period}{name={$\period$},description={Length of each period},
sort={L}}

\NewDocumentCommand{\simplex}{}{\bigtriangleup}

\NewDocumentCommand{\trun}{}{\debug{u}}
\NewDocumentCommand{\jobsrun}{}{\debug{n}}

\NewDocumentCommand{\regret}{}{\debug{R}}

\NewDocumentCommand{\tststrat}{O{\pl} O{\ti} O{\countvar}  }{\debug{x}^{#1}_{\ifthenelse{\equal{#2}{}}{#3}{#2, #3}}}

\NewDocumentCommand{\weight}{O{\ti} O{\pureb[\pl]} O{\pl}}  {\debug{w_{#1}^{#3}(#2)}}
\newglossaryentry{weight}{name={$\weight$}, description={Weight of action $\pureb[\pl]$ of player $\pl$ in the \ac{EWA} at time $\ti$}, sort={w}}

\NewDocumentCommand{\ww}{O{t} O{\pureb[\pl]} O{n}  }{\debug{w^{\pl}_{\ifthenelse{\equal{#1}{}}{#3}{#1, #3}}\sbraced{#2}}}
\newglossaryentry{mlweight}{
    name={$\ww$},
    description={Weight of action $\pureb[\pl]$ of player $\pl$ in level $n$ in the \ac{MLEWA} at time $\ti$},
    sort={ww}}

\NewDocumentCommand{\Ww}{O{\firsttime} O{\pureb[\pl]} O{\jobsrun}  }{\debug{W^{\pl}_{#1, #3}(#2)}}
\newglossaryentry{initweight}{
    name={$\Ww$},
    description={Initial weight of action $\pureb[\pl]$ of player $\pl$ in level $n$ in the \ac{MLEWA}},
    sort={W}
}

\NewDocumentCommand{\firsttime}{O{\pl} O{\jobsrun}}{\debug{\tau_{#2}^{#1}}}
\newglossaryentry{firsttime}{
    name={$\firsttime$}, description={First time that player $\pl$ has $n$ jobs to dispatch},
    sort={tau}}
    
\NewDocumentCommand{\smness}{}{\debug{\eta}}
\newglossaryentry{smoothness}{name={$\smness$}, description={Smoothness of weights updating in \ac{EWA}}, sort={eta}}
\NewDocumentCommand{\countvar}{}{\debug{n}}
\NewDocumentCommand{\countn}{O{\countvar} O{\ti} O{} }{\debug{\lambda}^{#3}_{#2, #1}}
\NewDocumentCommand{\countindn}{O{\countvar} O{\ti} O{\pl} }{\debug{\mu}^{#3}_{#2, #1}}

\newglossaryentry{countn}{name={$\countn$},description={Number of times between 1 and $\ti$ that $\tent$ is equal to $\countvar$},
sort={lambda}}

\NewDocumentCommand{\sm}{ }{\smness}
\NewDocumentCommand{\vx}{O{\ti}  }{\debug{x}^{\pl}_{ #1, \countvar}}

\NewDocumentCommand{\costz}{ }{\tcost[0,\pureprof[-\pl]]}
\NewDocumentCommand{\costb}{ }{\tcost[\pureprof]}

\NewDocumentCommand{\myb}{ }{\purep[\pl]}

\NewDocumentCommand{\ie}{}{i.e., }

\NewDocumentCommand{\pureprof}{O{}}{\purep[#1]} 
\NewDocumentCommand{\mixedprof}{O{}}{\mixedp[#1]} 
\NewDocumentCommand{\tmixedprof}{O{\ti}}{{\mixedp}_{#1}}

\NewDocumentCommand{\indic}{O{}}{\mathbf{1}_{\{#1\}}}

\NewDocumentCommand{\kz}{ }{\debug{k_{0}^{-\pl}}}

\NewDocumentCommand{\tkz}{ }{\debug{k_{\ti}^{-\pl}(0)}}

\NewDocumentCommand{\maxmove}{}{M}
\NewDocumentCommand{\nti}{}{\ti+1}
\NewDocumentCommand{\tib}{}{t'}
\NewDocumentCommand{\walk}{O{\ti}}{X_{#1}}
\NewDocumentCommand{\sigwalk}{O{\ti}}{\mathcal F_{#1}}
\NewDocumentCommand{\occl}{}{m}
\NewDocumentCommand{\occ}{O{\ti} O{\lev}}{o_{#1, #2}}
\NewDocumentCommand{\col}{}{\pl}
\NewDocumentCommand{\countr}{O{\ti} O{\col} O{x}}{\debug{\mu_{#1, #2, #3}}}
\NewDocumentCommand{\princr}{O{\lev} O{\occl}}{r(#1, #2)}
\NewDocumentCommand{\maxpr}{}{\rho}
\NewDocumentCommand{\sumpr}{}{A}
\NewDocumentCommand{\lev}{}{x}
\NewDocumentCommand{\levz}{}{z_0}

\DeclarePairedDelimiter{\braces}{\{}{\}}

\DeclarePairedDelimiter{\parens}{(}{)}
\DeclarePairedDelimiter{\abs}{\lvert}{\rvert}

\DeclarePairedDelimiterX{\braket}[2]{\langle}{\rangle}{#1,#2}

\DeclarePairedDelimiterX{\inner}[2]{\langle}{\rangle}{#1,#2}
\DeclarePairedDelimiterX{\setdef}[2]{\{}{\}}{#1:#2}

\DeclarePairedDelimiterXPP{\probof}[1]{\Prob}{(}{)}{}{%

#1}

\DeclarePairedDelimiterXPP{\exof}[1]{\Expect}{[}{]}{}{%

#1}

\newcommand{\initializationnormalized}{\tststrat[\pl][\firsttime][\countvar]}

\begin{document}
\title{Strategic Behavior and No-Regret Learning in Queueing Systems}

\author{Lucas Baudin%
\thanks{LAMSADE, Université Paris-Dauphine, Paris, France. \textsf{lucas.baudin@dauphine.eu}}
\and
Marco Scarsini%
\thanks{Dipartimento di Economia e Finanza, Luiss University, Viale Romania 32, 00197 Rome, Italy 
\textsf{mscarsini@luiss.it}, https://orcid.org/0000-0001-6473-794X}
\and
Xavier Venel%
\thanks{Dipartimento di Economia e Finanza, Luiss University, Viale Romania 32, 00197 Rome, Italy,  
\textsf{xvenel@luiss.it}, https://orcid.org/0000-0003-1150-9139} 
}

\maketitle

\begin{abstract}

This paper studies a dynamic discrete-time queuing model where at every period players get a new job and must send all their jobs to a queue that has a limited capacity. Players have an incentive to send their jobs as late as possible; however if a job does not exit the queue by a fixed deadline, the owner of the job incurs a penalty and this job is sent back to the player and joins the queue at the next period. Therefore, stability, i.e. the boundedness of the number of jobs in the system, is not guaranteed. We show that if players are myopically strategic, then the system is stable when the penalty is high enough. Moreover, if players use a learning algorithm derived from a typical no-regret algorithm (exponential weight), then the system is stable when penalties are greater than a bound that depends on the total number of jobs in the system.

\end{abstract}

\section{Introduction}

In the classical treatment of queues agents arrive at random times, wait according to a specified regime, and then get served; the service time is also random. 
In these models there is no room for any strategic behavior of the agents. Even when agents balk or renege, this is modeled as a random event, not a strategic choice of the agents.
Starting with the seminal paper by \citet{Naor:1969}, strategic elements have been included in queueing models. 
For instance, in Naor's model, when agents arrive, they rationally choose whether to join the queue or to balk.

The suitable tools to analyze strategic queueing models come from game theory.
The literature has considered several strategic models of queueing systems under different stochastic assumptions, different service regimes, and different strategy sets for the players. 
Various goals have been considered, such as computing Nash equilibria of the games, studying their efficiency, and examining the system stability under various equilibria.

One interesting class of problems, first examined by \citet{GlaHas:EJOR1983}, deals with situations where players need to be serviced before a fixed deadline, and otherwise pay a steep penalty.
In several variations of this model, players play the game repeatedly, for instance, they commute daily to the office and need to get there on time.

In a recent development, \citet{GaiTar:EC2020,GaiTar:EC2021}
considered a discrete-time queueing model where agents use learning algorithms to make the decision of which server to choose. 
The novelty of their analysis was the consideration of spillovers from one period to the other.
One of their goals is to establish conditions for the system to be stable.

\subsection{Our contribution}
\label{suse:contribution}

Our paper draws both on the literature on queues with a fixed deadline and on the contribution on learning and studies a discrete-time model where agents at every period receive a new job that requires service from a single server and need to decide when their jobs join a queue, taking into account a trade-off between waiting costs and a stiff penalty for being late. 
The model includes spillovers, since the jobs that cannot be served by the deadline go back to their owners and are to be sent to the server the following period; these late jobs are then added to the incoming daily new job.
From a queuing point of view, the model is deterministic: each player gets exactly one new job at every period. 
Randomness is due to the actions of the players, which can be mixed, and to the regime of the queue: if several jobs join the queue at the same time, the order in which they get served is uniformly random.

We consider several aspects of the model under the assumption that agents are strategic, but myopic, \ie at every period they play an equilibrium, but do not take into account the future effect of their actions.
This leads us to examine the equilibria of the single-period game. 
In this framework we show that the structure of equilibria depends on whether the number of jobs in the system is or is not larger than the number of times in each period; 
moreover it depends on whether the penalty cost for being late is or is not large enough.

When the number of jobs in the system exceeds the number of times in each period, and the penalty cost is large enough, we show that the stage game has a single \acl{CCE} (hence, a single \acl{NE}), where all players sends all their jobs to the queue as early as possible.
When the number of jobs in the system does not exceed the number of times in each period, then the stage game has multiple equilibria, whose structure we study.
Surprisingly, even if each player has a minmax strategy that guarantees that this player's jobs will meet the deadline, nevertheless, in equilibrium some jobs will be late with positive probability. 
This implies that the number of jobs in the system in the following period will be larger than in the current period.

In the second part of the paper we study a model where players use a no-regret learning algorithm to make their choices.
As we mentioned before, the number of jobs that each player may vary from one period to another. 
To face this, we will adopt a variation of the \acl{EWA} that takes into account the changing environment.

We describe the model using the language of game theory, but other motivations are possible.
For instance, we could consider a revenue-management interpretation where at each period agents buy a priority for their jobs. 
There are different priority levels and their price is monotone with the priority.
Jobs with high priority are served before jobs with lower priority, up to the fixed capacity for each period.
When there is a priority conflict, it is resolved at random.
The constraint is that a maximum of one agent with the lowest priority is served, a maximum of two agents with the lowest two priorities are served, etc.

\subsection{Related literature}

The analysis of strategic behavior in queueing systems goes back to the seminal paper of \cite{Naor:1969}, who studied an M/M/1 queue with a \ac{FIFO} policy where the agents' payoff consists of the reward that they get when they get served, minus a waiting cost that is proportional to the time they spend in the queue.
Once they arrive, they can decide whether to join the queue or to balk.
If they play a Nash equilibrium, their behavior is socially inefficient, in the sense that it does not maximize the social payoff (the sum of all players' payoffs). 
This is due to the fact that one agent's selfish behavior does not take into account the externalities it creates on other agents.
\citet{Has:E1985} showed that optimality can be achieved by a \ac{LIFO} policy.
The literature on strategic queueing system has then exploded. 
The reader is referred to \citet{HasHav:Kluwer2003} and \citet{Has:CRC2016} for a general treatment of the topic.

Some models have considered strategic agents who can decide when to join a queue.
The seminal paper by \citet{GlaHas:EJOR1983} considered a model, called ?/M/1, where agents arrive at a facility that every day starts service at time \(0\), serves all customers that arrive by some time \(T\) according to a \ac{FIFO} policy.
Each day, agents decide whether to visit the facility or not. 
If they do, they pick their arrival time with the goal of minimizing their expected time in the queue. 
This queueing literature with a strategic choice of the arrival time has been recently surveyed by
\citet{Haviv_Ravner:2021}.

In the framework of transportation theory, \citet{Vic:AER1969} studied a bottleneck model where agents choose the time they leave home to reach the office and used some fluid approximation.
\citet{RivScaToM:SSRN2018} studied a discrete version of the Vickrey model and examined Nash and correlated equilibria and their efficiency. Our model is a (variation of) a repeated version of their model with spillover. \citet{KawKonYuk:IJGT2023} extended the study of this bottleneck model to more general preferences and provided  conditions for the existence of a pure Nash equilibrium; moreover they examined the link with strong equilibria.

Recently, \citet{GaiTar:EC2020} proposed a model where several queues receive packets with a fixed, time-independent (but queue dependent) probability and must send them to a number of servers. 
Each server may process a packet it received with a fixed, time-independent probability. 
The paper studied the behavior of the system when players use no-regret learning procedures and in particular determined the conditions on the model parameters for the system to be stable. 
In a subsequent paper, \citet{GaiTar:EC2021} compared the behavior of no-regret, short-term, learners with players who adopt a long-run optimizing behavior. 
\citet{sentenac2021decentralized} used a similar model and introduced a new cooperative and decentralized algorithm that players can use. 
Compared to this work, our paper does assume that players intend to cooperate, in the sense that they use a standard algorithm to maximize their own payoff.

No-regret learning procedures algorithms are a family of online reinforcement learning algorithms such that they are at least as good as any constant action selected in hindsight (\emph{i.e.,} given that other players or the environment actions are fixed). The notion, also called Hannan or universal consistency \cite{fudenbergTheoryLearningGames1998}, was originally introduced by \citet{hannanApproximationBayesRisk1957} and is satisfied by multiple algorithms~\citep[see, for a recent review,][]{perchetApproachabilityRegretCalibration2014}. In particular, the \acf{EWA}~\cite{littlestoneWeightedMajorityAlgorithm1994, cesa-bianchiPredictionLearningGames2006} is a simple but efficient reinforcement learning procedure which adjusts weights of actions based on past experiences. However, there is no widely accepted extension of the concept of no-regret learning to systems with (endogenously changing) state variables (for instance Markov Decision Processes) such as ours (the state of the queue), so we propose a simple, multi-level, extension to \ac{EWA}.
\citet{BesGurZee:OR2015,BesGurZee:SS2019} dealt with stochastic optimization and single-agent, multi-armed bandit problems with temporal uncertainty in the rewards. 
In a recent interesting paper, \citet{CriGurLig:arXiv2022} studied---in the context of (non-stochastic) repeated games---strategies that have no-regret compared to dynamic sequences of actions whose number of changes scales sublinearly in time. There is also a large literature on time-varying games were studied by \citet{CarAbeWanXu:ICML2019,MerSta:IEEECDC2021,DuvMerStaVer:MOR2022,ZhaZhaLuoZho:ICML2022,AnaPanFarSan:arXiv2023}, among others. In a game theoretic framework, the reward functions of a player can change across time for two reasons: a change in the game and a change in the strategy of the other players. In these articles, the authors investigate different solution concepts to take this distinction into account. A key difference with our model is that the change of the game is exogeneous whereas in our model it is endogeneous. We restrict ourselves to the more classical notion of regret.

\subsection{Outline}
\label{suse:outline}
\cref{sec:model} introduces the model and the key concept that the paper will consider. 
\cref{sec:myopic} analyzes the behavior of myopic strategic players. 
The key point in this section is the analysis of the one-shot game for any job allocation to the players. 
\cref{sec:learning}
analyzes the repeated game  with learning players.

\section{The model}
\label{sec:model}

We consider a discrete-time game where \(\players\) is the set of players and time is split into periods of equal length \(\period\ge 2\). 
The generic period is denoted by \(\ti\).
Throughout the paper we assume that there are at least \(\abs{\players} \ge 3\) players.

At the beginning of every period \(\ti\ge 0\) each player \(\pl\) receives one job and independently chooses an action \(\tact \in \braces*{0,\dots,\period-1}\) \glsadd{timedaction}, which represents the time during period \(\ti\) at which \(\pl\)'s jobs join the queue.
The symbol \(\pureprof_{\ti} \coloneqq \parens{\tact}_{\pl\in\nplayers}\)\glsadd{purep} denotes the action profile at period \(\ti\).
Since all of player~$\pl$'s jobs  join the queue at the same time, the action set $\actions$ does not depend on the number of jobs held by each player.

The choice of \(\tact\) may depend on the past history of the game.
Depending on the actions taken by all players, some jobs are late, \ie they cannot exit the queue before the end of the period.
Late jobs are returned to their respective owners and will join the queue once more at the following period.
This produces a \emph{spillover} effect: at every period, the number of jobs that player \(\pl\) handles is the number of \(\pl\)'s late jobs in the previous period plus one.

The state of the system at period \(\ti\) is the vector \(\st=(\tent)_{\pl\in \players}\), where  \(\tent\) is the number of player \(\pl\)'s jobs at the beginning of period \(\ti\).
The state space is \(\states \coloneqq \N^\players\)\glsadd{states}.
The total number of jobs at time \(\ti\) is 
\begin{equation}
    \label{eq:tot-jobs-t}
    \tnent\coloneqq \sum_{i \in \nplayers} \tent.
\end{equation}

Late jobs at period \(\ti\) incur a penalty cost \(\tplty\) that depends on the total number \(\tnent\) of jobs in the system at period \(\ti\).
At the end of period \(\ti\), each player \(\pl\) pays a cost \(\tcost\) that is the sum of two components: the waiting cost, \ie the time spent in the queue by each of \(\pl\)'s jobs, and the penalty cost:
\begin{equation}
    \label{eq:cost-static}
    \tcost[\pureprof_{\ti}] = \tent (\period -\tact) + \tplty \E[\text{number of $\pl$'s jobs that are late at \(\ti\)}].
\end{equation}

Given an action profile \(\pureprof_{\ti}\), all of player \(\pl\)'s jobs are assumed to have the same probability of being late.
As a consequence
\begin{equation}
    \label{eq:expected-late}   
    \E[\text{number of $\pl$'s jobs that are late at $\ti$}]=\tent \late{\pureprof_{\ti}},
    \end{equation}
    which implies
    \begin{equation}
    \label{eq:cost-static-2}
    \tcost[\purep_{\ti}] = \tent (\period -\tact) + \tplty \tent \late{\purep_{\ti}}.
\end{equation}

A \acfi{DQM}\acused{DQM} is specified by a period length $\period$\glsadd{period}, a number of players $\players$ and a sequence of penalties $(\tplty)_{\ti \in \N}$\glsadd{penalty} that depend on the total number \(\tnent\) of jobs in the system.

Given the spillover effect, it is natural to investigate the asymptotic behavior of the number of jobs $\tnent$. 
In particular, if \(\tnent > \period\), at period \(\ti\) there are more jobs than time units; therefore, at least one job is necessarily late. 
Some jobs can be late even if \(\tnent \leq \period\). 
In the rest of the paper, we will be interested in bounding the number of jobs in the system, under various players' behaviors. 
More formally, we let \(\histset_{\ti} \coloneqq \parens*{\states\times\actions}^{\ti-1 }\times\states\) denote the set of histories of length \(\ti\) and \(\histset\coloneqq\cup_{\ti\ge 1} \histset_{\ti}\) the set of finite histories. 

\begin{definition}[Strategy]\label{de:strategy} 
     A \emph{strategy} is a function \(\strat \colon \histset \to \simplex(\actions)\)\glsadd{strat}. 
    When there is no risk of ambiguity, we will let \glsadd{tstrat}\(\tstrat \coloneqq \strat(\hist)\) denote player~\(\pl\)'s  mixed action at period \(\ti\).
\end{definition}

\begin{definition}[Stability of a strategy profile]\label{de:stability} 
    A strategy profile $\strat[]$ in a \ac{DQM} is said to be \emph{stable} if the number of jobs $\tnent$ is almost surely bounded, \ie
    \[\mathbb P_{\strat[]}\left[\exists M, \forall \ti \in \N, \tnent \leq M\right] = 1.
    \]
\end{definition}

The following assumption will be in place throughout the paper.
\begin{assumption}\label{hyp:rep}
    The following inequality holds: \(\period \geq \nplayers\).
\end{assumption}

Without \cref{hyp:rep}, no strategy profile in a \ac{DQM} could be stable, because at every period at most \(\period\) jobs leave the system, and \(\nplayers\) new jobs arrive. 

We study the stability of two different families of strategy profiles. 
In the first case, agents are myopically strategic, \ie they focus on the situation at the current period and play in a strategic way for that period, \ie 
they play at every period a correlated equilibrium.
We prove that the \ac{DQM} is stable under this strategy profile, provided the penalty costs $\tplty$ are large enough (potentially constant). 
In the second case, agents use no-regret strategies that depend on their number of jobs. 
We provide sufficient conditions for the stability of these strategies in the \ac{DQM}.

\section{Myopic strategic players}
\label{sec:myopic}

In this section, players are assumed to be strategic and myopic. The myopic assumptions implies that at period \(\ti\) players only consider the costs that they may pay at the current period without taking into consideration the effect that their actions have on the future of the game. 
In other words, players repeatedly play one-shot games with payoff functions $\parens{\tcost[][\pl][\st]}_{\pl \in \players}$.

We now introduce the definitions of myopic \acfi{NE}\acused{NE} and myopic \acfi{CCE}\acused{CCE}.

\begin{definition}[Myopic \ac{NE}]\label{de:strategic-profile}
    A strategy profile \(\strat[] \colon \histset \to \simplex(\actions^\players)\) is called a \emph{myopic \ac{NE}} if, for any history $\hist$\glsadd{history} with current state $\st$, the mixed action $\strat[](\hist)$ is a \acl{NE} of the one-shot game with jobs \(\st = (\tent)_{\pl\in \players}\).
\end{definition}

\begin{definition}[Myopic \ac{CCE}]\label{de:strategic-profile-coarse}
    A strategy distribution \(\strat[] \colon \histset \to \simplex(\actions^\players)\) is called a \emph{myopic \ac{CCE}} if, for any history $\hist$\glsadd{history} with current state $\st$, the distribution $\strat[](\hist)$ is a \acl{CCE} of the one-shot static game with jobs \(\st = (\tent)_{\pl\in \players}\).
\end{definition}

When the number of jobs is lower than \(\period\), any player who plays \(0\) is sure not to pay the penalty cost. 
Even in this case, there exist equilibria where some late jobs are late with positive probability. 
Therefore, the number of jobs may not be constant throughout the play, as detailed in \cref{suse:T>=k}.

We now analyze the family of one-shot games parameterized by the number of jobs owned by each player.

The rest of the section is organized as follows. 
We start by focusing on the one-shot game, first when  \(\period \ge \tnent\), and then 
when \(\period < \tnent\).
We conclude by analyzing the global queue.

\subsection{Equilibria in the one-shot game when \texorpdfstring{\(\period \geq \tnent\)}{T > k}}
\label{suse:T>=k}

We will show that, even when \(\period \geq \tnent\), in equilibrium some jobs are late with positive probability. 

We first prove that only the last \(\tnent\) actions are used in an equilibrium; hence, it is enough to consider the case  \(\period=\nent\).

\begin{lemma}[Dominated actions for \(\period > \tnent \)]
    \label{lem:static:domin}
    If \(\period > \tnent\), then for each player \(\pl\), any action \(\tact \in \braces{0,\dots, \period-\tnent-1}\) is strictly dominated by action \(\period-\tnent\). 
    Furthermore, for any  (mixed) action profile~\(\tmixedprof\)\glsadd{mixedp}, 
    \begin{equation}
    \label{eq:cost-T-k}      
    \tcost[\period-\tnent, \tmixedprof^{-\pl}] = \tent \tnent.
    \end{equation}
\end{lemma} 

\cref{eq:cost-T-k} shows that the cost \(\tent \tnent\) that player \(\pl\) pays when playing action $\period-\tnent$ is  independent of the other players' actions. 

\begin{proof}[Proof of \cref{lem:static:domin}]
    If player \(\pl\)'s jobs join the queue at time \(\tact\), then they leave the         queue no later than time \(\tact+l\), where \(l\) is the number of jobs that the queue at time \(\tact\) or before. 
    Since \(l\leq \tnent\), player \(\pl\)'s jobs  leave the queue no later than time \(\tact+\tnent\), which is smaller or equal than \(\period\) if \(\tact\) is smaller or equal than \(\period-\tnent\). 
    So, the probability \(\late{\tact, \tmixedprof^{-\pl}}\) that a specific job of player \(\pl\) is late is \(0\) and player \(\pl\)'s cost is \(\tent(\period-\tact)\).
    In particular, the cost of choosing action \(\period-\tnent\) is \(\tent \tnent\), which is strictly smaller than the cost of any \(\tact \in \braces*{0, \dots, \period-\tnent-1}\).
\end{proof}

\cref{{lem:static:domin}} implies that it is enough to consider the case  \(\period=\tnent\).
In the rest of this section, all the results will be proved under this hypothesis. 
The general case \(\period\ge\tnent\) can be obtained by renaming the actions.
  
The next theorem shows that, in any equilibrium, players can be divided into two groups: 
players in the first group choose the same mixed action and mix only on the two first non-dominated actions \(\period-\tnent\) and \(\period-\tnent+1\); 
players in the second group do not mix on \(\period-\tnent\) and put strictly positive weight on \(\period-\tnent+1\). 
The way they mix on the remaining non-dominated actions is not specified.

The symbol \(\socialcost\) \glsadd{socialcost} denotes the social cost of a profile \(\tmixedprof\), that is,
\begin{equation}\label{eq:SC}
  \socialcost\coloneqq\sum_{\pl\in\nplayers} \tcost.
\end{equation}
  
\begin{theorem}[Structure of Nash equilibria]\label{thm:structure-Nash}
    If \(\period \geq \tnent\), \(\tplty > \tnent^2\), and \(\tmixedprof\) is a Nash equilibrium, then:
    \begin{enumerate}[label=\emph{(\roman*)}, ref=(\roman*)]
    \item 
    \label{it:thm:structure-Nash-1}
    for each player \(\plb\) and action \(\tact[\ti][\plb] < \period - \tnent\), we have $\tmixed[\ti][\plb](\tact[\ti][\plb]) = 0$;
         
    \item
    \label{it:thm:structure-Nash-2}
    there exists a player \(\pl\) such that
        \begin{enumerate}[label=\emph{(\alph*)}, ref=(\alph*)]
        \item 
        \label{it:thm:structure-Nash-2-a}
        \(\tmixed(\period-\tnent) > 0\),
              
        \item
        \label{it:thm:structure-Nash-2-b}
        for all \(\tact > \period-\tnent+1\), we have \(\tmixed(\tact) = 0\);
        \end{enumerate}
    
    \item 
    \label{it:thm:structure-Nash-3}
    for every player \(\plb \neq \pl\), \ \(\tmixed[\ti][\plb](\period-\tnent+1) > 0\);
        
    \item 
    \label{it:thm:structure-Nash-4}
    for every player \(\plb\), if \(\tmixed[\ti][\plb](\period-\tnent) > 0\), then \(\tmixed[\ti][\plb] = \tmixed[\ti][\pl]\);
         
    \item 
    \label{it:thm:structure-Nash-5}
    \(\tnent^2-\tnent+1 \leq \socialcost \leq \tnent^2\).
    \end{enumerate}
\end{theorem}

\begin{proof}
    \ref{it:thm:structure-Nash-1}
    \emph{The actions before \(\period-\tnent\) are not chosen.} 
    By \cref{lem:static:domin}, we know that all actions smaller than \(\period-\tnent\) are strictly dominated. 
    The result follows from the fact that Nash equilibria do not mix on dominated actions. 
    In the rest of the proof, we will assume that \(\period=\tnent\). 
    As mentioned before, if \(\period<\tnent\), then the proof can be easily adapted by translation: action \(0\) becomes \(\period - \tnent\), \(1\) becomes \(\period - \tnent +1\), etc.
    
    \ref{it:thm:structure-Nash-2}\ref{it:thm:structure-Nash-2-a}
    \emph{There exists a player \(\pl\) such that \(\tmixed(0) > 0\).}
    Assume \emph{ad absurdum} that \(\tmixed[\ti][\plb](0) = 0\) for each player~\(\plb\). 
    Then  at least one job is late. 
    This implies that 
    \begin{equation}\label{eq:E-late-jobs}
        \E[\text{number of late jobs}] = \sum_{\plb \in \players} \tent[\ti][\plb] \late[\plb]{\tmixedprof} \geq 1
    \end{equation}
    and there exists \(\pl\) such that \(\tent \late{\tmixedprof} \geq \tent/\tnent\). 
    Then, \[\tcost \geq \tent \tplty \late{\tmixedprof} \geq \tplty\tent/\tnent.\] 
    So if       \(\tplty > \tnent^2\), player \(\pl\)'s cost is strictly greater than \(\tent \tnent\), which is a contradiction. 
    Therefore, one player \(\pl\) mixes on \(0\) and       \(\tcost = \tent \tnent\).
    
    \ref{it:thm:structure-Nash-3}
    \emph{Every other player mixes on \(1\).}
    If this is not the case, then \(\tmixed[\ti][\plb](1)=0\)
    for some \(\plb \neq \pl\). 
    We now prove that \(\pl\) can profitably deviate by playing the pure action \(1\).
    Assume that, indeed, \(\pl\) chooses action \(1\).
    If \(\plb\) plays the pure action \(0\), then player \(\pl\)'s jobs joining the queue at \(1\) are not late because the number of jobs joining the queue at \(1\) is smaller than \(\tnent=\period\), so they are processed before the end of the period. 
    Else, if \(\plb\) does not play \(0\), then \(\plb\)'s action is not smaller than \(2\). 
    Then the number of jobs that join the queue at \(0\) or \(1\) is at most \(\period-1\) and these jobs are processed before the end of the period. 
  
    So, by deviating to \(1\), \(\pl\)'s jobs are not late and this leads to a cost \(\tcost[1, \tmixedprof^{-\pl}]) = \tent (\tnent-1)\). 
    Hence, the deviation is profitable and \(\tmixedprof\)  cannot be an equilibrium.
    
    \ref{it:thm:structure-Nash-4}    
    \emph{(First step) All players who mix on \(0\) mix on \(1\), and put the same the weight on action \(1\).}
    Suppose that at least two players mix on \(0\). 
    We first prove that they mix on \(1\) too. 
    Call these two players \(\pl\) and \(\plb\).
    Player \(\pl\) satisfies \ref{it:thm:structure-Nash-2}\ref{it:thm:structure-Nash-2-a}; therefore, by \ref{it:thm:structure-Nash-3}, player \(\plb\) mixes on \(1\).
    Symmetrically, player \(\plb\) satisfies \ref{it:thm:structure-Nash-2}\ref{it:thm:structure-Nash-2-a}; therefore, by \ref{it:thm:structure-Nash-3}, player \(\pl\) mixes on \(1\).
    This shows that both players mix on \(1\).
    
    We now prove that they put the same weight on \(1\).
    A job sent by any of these players can be late only if all jobs join the queue at time \(1\).
    Indeed, if at least one job joins the queue at \(0\), then none of the jobs that join the queue at \(1\) can be late because there are \(\period-1\) units of time to process them.
    
    Therefore, a job departing at \(1\) is late only if all other jobs also depart at \(1\), i.e.,  all players play action \(1\). 
    If this happens, then the probability that a job is late is \(1/\tnent\). 
    Therefore, the probability that a job owned by player \(\pl\) is late, conditionally on the fact that  \(\pl\)  plays \(1\), is equal to
    \begin{equation*}
        \late{1, \tmixedprof^{-\pl}} = \frac{\proba[\tact[\plb]=1\ \forall \plb \neq \pl]}{\tnent} = \frac{\prod_{\plb\neq \pl} \tmixed[\ti][\plb] (1)}{\tnent}.
    \end{equation*}
    
    Furthermore, since in a Nash equilibrium player \(\pl\)  is indifferent between \(0\) and \(1\), we have 
    \begin{equation}
    \label{eq:cost-indiff}
        \tcost[0,\tmixedprof^{-\pl}]
        = \tent \tnent 
        = \tcost[1,\tmixedprof^{-\pl}] 
        = \tent\parens*{\tnent-1+\frac{\prod_{\plb\neq \pl} \tmixed[\ti][\plb] (1)}{\tnent}\tplty}.
    \end{equation} 
    Thus, we have 
    \begin{equation}
    \label{eq:prob-penalty}
        \prod_{\plb\neq \pl} \tmixed[\ti][\plb](1) \tplty = \tnent.
    \end{equation}
    This applies to any
    \(\plb\) such that \(\tmixed[\ti][\plb](0) > 0\).
    Hence \(\tmixed[\ti][\pl](1) = \tmixed[\ti][\plb](1)\).
    
    \ref{it:thm:structure-Nash-2} \ref{it:thm:structure-Nash-2-a} and \ref{it:thm:structure-Nash-4}
    \emph{(second step) 
    Players who mix on \(0\) do not mix on any action strictly larger than \(1\).}
    If player \(\pl\) mixes on an action \(\tact > 1\), then one of \(\pl\)'s jobs is late whenever all other players play  \(1\). 
    This implies that the probability to be late playing action \(\tact > 1\) is greater than the probability that all other players play \(1\), so using~\eqref{eq:prob-penalty} we get
    \begin{equation*}
        \tcost[\tact, \tmixedprof^{-\pl}] 
        \geq \tent\braces*{\tnent - \act + \prod_{\plb\neq \pl}\tmixed[\ti][\plb](1)\tplty} 
        = \tent \tnent + \tent (\tnent - \tact) > \tent \tnent.
    \end{equation*}
    Since all players who mix on \(0\) only mix between \(0\) and \(1\) with the same weight on \(1\), they actually play the same mixed action, which proves \ref{it:thm:structure-Nash-4}.
    
    \ref{it:thm:structure-Nash-5}     \emph{Social cost.}
    First we prove  that the social cost is smaller than \(\tnent^2\). 
    If \(\tmixedprof\) is a Nash equilibrium, then for every player $\plb$, we have
    \begin{equation}
    \label{eq:Nash-eq-cond}    
        \tcost[\tmixedprof][\plb] \leq \tcost[0, \tmixedprof^{-\plb}][\plb] = \tent[\plb]\tnent.
    \end{equation}
    Then, 
    \(\socialcost \leq \sum_{\plb \in \players} \tent[\plb] \tnent = \tnent^{2}\).
    We now prove that the social cost is greater than \(\tnent^2 - \tnent + \tent\). 
    One player \(\pl\) mixes on \(0\) and incurs a cost equal to \(\tent \tnent\), and all other players 
    \(\plb \neq \pl\) mix on \(1\), so their cost is at least \(\tent[\plb](\tnent-1)\). 
    Thus,
    \begin{equation*}
    \label{eq:SC-cond}
        \socialcost \geq (\tnent-\tent)(\tnent -1) + \tent\tnent = \tnent^2 - \tnent + 1. \qedhere
    \end{equation*}
\end{proof}

\begin{corollary}
    For any Nash equilibrium $\mixedp$, there is positive probability that a job is late, that is 
    \[
    \mathbb E[\emph{number of late jobs}] > 0.
    \]
\end{corollary}
  
\begin{proof}
  By \cref{thm:structure-Nash} \ref{it:thm:structure-Nash-2} \ref{it:thm:structure-Nash-2-a}, we know that there exists a player $\pl$ who mixes on $\period-\tnent$ with positive probability. 
  We consider two cases, based on this probability. 
  If $\pl$ plays with probability strictly less than \(1\), then by \cref{thm:structure-Nash} \ref{it:thm:structure-Nash-3} and \cref{thm:structure-Nash} \ref{it:thm:structure-Nash-4} we know that all the players are playing simultaneously in $\period-\tnent+1$ with positive probability. 
  Therefore, one of them is late.

  Let us assume that $\pl$ plays with probability \(1\) the action $\period-\tnent$. 
  If two players mix on $\period-\tnent$, we can apply  \cref{thm:structure-Nash} to each of them and therefore by \ref{it:thm:structure-Nash-3}, both put positive weight on \(\period-\tnent+1\). 
  So, \(\pl\) is the only player playing \(\period-\tnent\) and by \cref{thm:structure-Nash} \ref{it:thm:structure-Nash-4} every other player plays \(\period-\tnent+1\) with positive probability. 
  Let \(\plb \neq \pl\) and suppose that \(\tmixed[\ti][\plb](\period-\tnent+1) = 1\). 
  Then, there exists another \(\plb'\) different from both \(\pl\) and \(\plb\) (because \(\players \ge 3\)). 
  Player \(\pl\) plays \(\period-\tnent\), \(\plb\) plays \(\period-\tnent+1\), so \(\plb'\) is not late when it plays \(\period-\tnent+2\). 
  Therefore, \(\plb'\) strictly prefers \(\period-\tnent+2\) to \(\period-\tnent+1\), so it cannot put positive weight on \(\period-\tnent+1\), which is a contradiction.
\end{proof}

\subsection{Equilibria when \texorpdfstring{\(\period < \tnent\)}{T < k}}

We now study the game when the period length is smaller than the number of jobs, i.e., \(\period < \tnent\).

\begin{theorem}[\Aclp{CCE}]
\label{thm:cce_more_players}
    If \(\tnent > \period\)  and \(\tplty > \tnent^{2} \period \), then there is a unique \acl{CCE}, which is actually a pure equilibrium where all players play \(0\).
\end{theorem}

To show that the support of any \ac{CCE} is the singleton \(\boldsymbol{0}\), we proceed by contradiction. 
We assume the existence of a \ac{CCE} whose support is not \(\boldsymbol{0}\); we sum the cost of unilateral deviations to  \(0\) for each player and show that the sum is negative. 
As a consequence, there is at least one deviation cost which is negative, which shows that this player has a profitable deviation.
The proof is not ``constructive'' in the sense that the dissatisfied player is not designated outright.

  \begin{proof}[Proof of \cref{thm:cce_more_players}]
  Let \(\tcorreq\) be a \ac{CCE}. 
  We have
\begin{equation}
\label{eq:cost-CCE}
\tcost[\tcorreq] 
= \sum_{\pureprof\in \actions^{\players}} \tcorreq(\pureprof)  \tcost[\pureprof] = \sum_{\pureprof\in \actions^{\players}} \tcorreq(\pureprof) \tent \braces*{\period -\purep[\pl]+\late{\purep}\tplty}.
\end{equation}
  Since \(\tcorreq\) is a coarse correlated equilibrium, the cost that \(\pl\) obtains by a unilaterally deviating to \(0\) is not smaller then the cost that  \(\pl\)  incurs under $\tcorreq$, that is,
  \begin{equation}
  \label{eq:cost-dev-i}
    \sum_{\pureprof\in \actions^{\players}} \tcorreq(\pureprof) \tent \braces*{\period + \late{0, \purep[-\pl]}\tplty} 
    \geq 
    \sum_{\pureprof\in \actions^{\players}} 
    \tcorreq(\pureprof) \tent \braces*{\period - \purep[\pl]+\late{\pureprof} \tplty},
  \end{equation}
  which leads to
  \begin{equation}
  \label{eq:cost-dev-i-implic}
    \sum_{\pureprof\in \actions^{\players}} 
    \tcorreq(\pureprof) \tent \braces{\late{0, \purep[-\pl]}\tplty - \purep[\pl]+\late{\pureprof} \tplty} \geq 0.
  \end{equation}

A player who was originally playing \(0\) is actually not deviating. 
Therefore, the term of the sum that corresponds to players who  play \(0\) are equal to zero. 
  Therefore, in \eqref{eq:cost-dev-i-implic} we can sum  over all players who do not play \(0\) and get
  \begin{equation}
    \label{eq:player-no-play-0}
    \sum_{\pureprof \in \actions^{\players} \setminus \{\boldsymbol{0}\}} \tcorreq(\pureprof)\braces*{\sum_{\pl\in \players} \tent (\late{{0, \pureprof[-\pl]}}-\late{\pureprof})\tplty + \tent \act} \geq 0.
  \end{equation}

 A job that joins the queue after time \(0\), has a higher probability of being late:
  \begin{equation}
    \label{eq:player-i-late}
    \late{\pureprof} \ge \late{0, \purep[-\pl]}
  \end{equation}
The following claim refines the above statement.
\begin{claim}
\label{cl:prob-i-late}
Given an action profile \(\pureprof\neq\boldsymbol{0}
\), if 
\begin{equation}
\label{eq:a-max}
\purep[\plb] = \max_{\pl\in\players}(\purep[\pl]),
\end{equation}
then
\begin{equation}
\label{eq:prob-late-j-max}
\late[\plb]{\pureprof} \ge \frac{1}{\numsim\tnent} + \late[\plb]{0,\pureprof[\plb]}, 
\end{equation}
where 
\begin{equation}
\label{eq:n-jobs-a-j}
\numsim \coloneqq \sum_{\pl \in \players \colon \act = \act[\plb]} \tent
\end{equation}
is the number of jobs that join the queue at time \(\act[\plb]\).
\end{claim}

\begin{proof}
First of all notice that in \eqref{eq:a-max} we have \(\purep[\plb]>0\).
Call \(\numpre\) the number of jobs that join the queue at time \(\act[\plb]\) and leave the system before the deadline \(\period\).
Since \(\period\le\tnent\), some jobs are late; more precisely, \(\numsim-\numpre\) jobs that join the queue at \(\purep[\plb]\) are late.
Then
\begin{equation}
\label{eq:p-jobs-late-a-j}
\late[\plb]{\pureprof} = \frac{\numsim-\numpre}{\numsim}.
\end{equation}
Moreover, \(\late[\plb]{0, \pureprof[-\plb]} \leq (\tnent-\period)/\tnent\), so
\begin{equation}
\label{eq:p-jobs-late-a-j-0}
\late[\plb]{\pureprof} - \late[\plb]{0,\pureprof[-\plb]} 
\geq \frac{\numsim-\numpre}{\numsim} - \frac{\tnent-\period}{\tnent}
= \frac{-\numpre\tnent + \period \numsim}{\numsim\tnent}.
\end{equation}
 However, the maximum number of jobs that can leave the system without being late is \(\period\).
 
Under the action profile \(\pureprof\),  \(\tnent-\numsim\) jobs join the queue strictly before time \(\act[\plb]\). 
To finish the proof, we consider several cases related to the value of \(\tnent-\numsim\):

  \begin{itemize}
  \item
  If \(0 < \tnent - \numsim \le \period\), no more than $\period - (\tnent-\numsim)$ jobs joining the queue at time $\act[\plb]$ can arrive on time, i.e., \(\numpre \leq \period - (\tnent-\numsim)\). 
  Using this inequality in \eqref{eq:p-jobs-late-a-j-0}, we obtain
  \begin{equation*}
  \label{eq:p-late-a-0}
    \begin{aligned}
    \late[\plb]{\pureprof} - \late[\plb]{0, \pureprof[-\plb]} & \geq \frac{-(\period - \tnent + \numsim)\tnent + \period \numsim}{\numsim\tnent}
    \\ & = \frac{\period}{\tnent} - \frac{\period-\tnent}{\numsim} - 1
    \\ & = \braces*{\period - \tnent}\braces*{\frac{1}{\tnent} - \frac{1}{\numsim}}
    \\ & = \braces*{\tnent - \period}\frac{\tnent - \numsim}{\tnent \numsim}
    \\ & \ge \frac{1}{\tnent \numsim},
    \end{aligned}
  \end{equation*}
  because in this case \(\tnent-\numsim \ge 1\). 
  
  \item If \(\tnent-\numsim = 0\), then all the jobs join the queue at a time greater than \(1\); therefore \(\numpre \leq \period-1\).

  Since in this case \(\tnent=\numsim\), after some simplifications \eqref{eq:p-jobs-late-a-j-0} becomes 
  \begin{equation}
  \label{eq:p-late-a-0-2nd-case}
    \late[\plb]{\pureprof} - \late[\plb]{0, \pureprof[-\plb]} \geq \frac{-\numpre+\period}{\tnent} \geq \frac{-(\period-1)+\period }{\tnent} = \frac{1}{\tnent} \geq \frac{1}{\tnent \numsim}.
  \end{equation} 
  
  \item if \(\tnent - \numsim > \period\), then no jobs can avoid being late and \(\numpre=0\). 
  Therefore, \eqref{eq:p-jobs-late-a-j-0} becomes
  \begin{equation*}
  \late{\pureprof} - \late{0, \pureprof[-\pl]} \geq \frac{\period}{\tnent} \geq \frac{1}{\tnent \numsim}.
  \end{equation*}
\end{itemize}
This concludes the proof of the claim.  
\end{proof}

\begin{claim}
\label{cl:-1/k}
For any \(\pureprof \neq \boldsymbol{0}\), we have
  \begin{equation}
  \label{eq:-1/k}
      -\frac{1}{\tnent} \geq \sum_{\pl\in \players} \tent (\late{{0, \pureprof[-\pl]}}-\late{\pureprof}).
  \end{equation}
\end{claim}
  
\begin{proof}
We split the players into two groups: players whose jobs join the queue at time \(\purep[\plb]\) as defined in \eqref{eq:a-max} and the remaining ones. 
We have
  \begin{equation}
  \label{eq:sum-k-late}
  \begin{split}
    \sum_{\pl\in \players} \tent (\late{{0, \pureprof[-\pl]}}-\late{\pureprof}) 
    &= \sum_{\pl \in \players \colon \act \neq \act[\plb]} \tent (\late{{0, \pureprof[-\pl]}}-\late{\pureprof}) \\ 
    &\quad+ \sum_{\pl \in \players \colon \act = \act[\plb]} \ent(\late[\pl]{{0, \pureprof[-\pl]}}-\late[\pl]{\pureprof}). 
  \end{split}
  \end{equation}
  The first term is nonpositive because of  \eqref{eq:player-i-late}
  and the second term is nonpositive because of (\ref{eq:prob-late-j-max}), leading to
  \begin{equation}
    \label{eq:sum-k-i-p-late}
    \sum_{\pl\in \players} \tent (\late{{0, \pureprof[-\pl]}}-\late{\pureprof}) \leq -\sum_{\pl \in \players \colon \act = \act[\plb]} \tent\frac{1}{\numsim\tnent}.
  \end{equation}

The inequality in \eqref{eq:sum-k-i-p-late} yields
  \begin{equation*}
      \sum_{\pl\in \players} \tent (\late{{0, \pureprof[-\pl]}}-\late{\pureprof}) \leq -\frac{\numsim}{\numsim\tnent} = -\frac{1}{\tnent},
  \end{equation*}
which proves the claim.
\end{proof}  

The combination of \cref{cl:-1/k} and \cref{eq:player-no-play-0} yields
\begin{equation}
\label{eq:0-sum-tau}
\begin{split}
    0 & \le \sum_{\pureprof \in \actions^\players \setminus \{\boldsymbol{0}\}} \tcorreq(\pureprof)
    \braces*{\sum_{\pl\in \players} \tent (\late{{0, \pureprof[-\pl]}}-\late{\pureprof})\tplty + \tent \act}
    \\
    0 & \le  \sum_{\pureprof \in \actions^\players \setminus \{\boldsymbol{0}\}} \tcorreq(\pureprof) \braces*{-\frac{1}{\nent}\tplty + \sum_{\pl \in \players} \tent \act}
    \\ 0 & \le \sum_{\purep \in \actions^\players \setminus \{0\}} \tcorreq(\pureprof) \braces*{-\frac{1}{\tnent}\tplty + \period \tnent}.   \end{split}
  \end{equation}
  
Under the assumptions of the theorem, \(\tplty > \nent^2 \period\), so 
\begin{equation}
\label{eq:1/k-again}
 -\frac{1}{\tnent}+\period \tnent < 0.        
\end{equation} 
Therefore, \eqref{eq:0-sum-tau} can hold only if \(\tcorreq(\pureprof) = 0\) for all \(\pureprof \in \actions^\players \setminus \{\boldsymbol{0}\}\), which implies that the support of \(\tcorreq\) is \(\boldsymbol{0}\).
\end{proof}

\subsection{Global queue}

The following theorem is a consequence of the previous results.

\begin{theorem}[Stability for Strategic Repetition]
    \label{thm:stratstab}
    Consider a \ac{DQM} with \(\players \leq \period\). 
    If \(\, \inf_{\ent[]} \tplty[\ent[]] > (\period +\players) ^2\period\), then any myopic \ac{CCE} is stable.
\end{theorem}

\begin{corollary}[Stability for Strategic Repetition]
    \label{cor:stratstab}
    Consider a \ac{DQM} with \(\players \leq \period\). If \(\, \inf_{\ent[]} \tplty[\ent[]] > (\period +\players) ^2\period\), then any myopic \ac{NE} is stable.
\end{corollary}

Notice that in particular if the penalty cost is a constant \(\plty>(\period+\players)^{2}\period \), then any myopic \ac{CCE} is stable. The corollary is an immediate consequence of the theorem since any Nash equilibrium is also  a coarse correlated equilibrium.
 
As will be clear in the proof of \cref{thm:stratstab}, the queue alternates between two possible regimes. 
The first regime corresponds to the case \(\tnent \leq \period\). In this regime, stage equilibria have a non-trivial structure and in equilibrium some jobs may be late.
As a consequence, the number \(\tnent\) of jobs in the system may oscillate over time. 
When this number overcomes the level \(\period\), the system enters the other regime, corresponding to \(\tnent > \period\). In this regime the only stage equilibrium is the pure profile \(\boldsymbol{0}\). 
Therefore, if \(\players=\period\), the number \(\tnent\) of jobs in the system stays constant; 
if \(\players < \period\), the number of jobs decreases until the system goes back to the first regime. 

\begin{remark}
If the penalty costs are small, there may exist a myopic NE that is not stable. 
For example, if for all \( \ent[] \in \N,\ \tplty[\ent[]] < 1\), the cost of preempting the other players is too large compared to the potential gain: 
playing at the last stage of the period is strictly dominating for each player. 
The unique myopic NE is for every player to wait the last stage of the period and there are \(\tnent-1\) late jobs. 
Hence,  \(\tnent\) is almost-surely unbounded.
\end{remark}

\begin{proof}[Proof of \cref{thm:stratstab}]
    Assume that \(\players \leq \period\) and for all \(\nent\), \(\tplty[\ent[]] > (\period+\players)^2\period\). We show by induction that for every \(\ti\), \(\ent[] \leq \period+ \players\).
    
    The results is true at period~\(1\). 
    We now show that it is true at every period. 
    From period \(\ti\) to period \(\ti+1\), we need to consider two cases. 
    If \(\tnent \le \period\), then \(\tnent[\ti+1] \le \period + \players\) because at the next period \(\players\) new jobs arrive and at most \(\tnent\) are late. 
    If \(\period < \tnent \le \period + \players\), then
    \begin{equation}
        \label{eq:C-k^2T}
        \tplty > (\period+\players)^{2}\period \geq \tnent^{2} \period.   
    \end{equation}
    
    By \cref{thm:cce_more_players} there is a unique stage \ac{CCE}, where all players play~\(0\).
    This implies that at least \(\period\) jobs leave the system.
    Since \(\players \le \period\) and there are \(\players\) new jobs in the next period, we have 
    \(\tnent[\ti+1] \le \tnent \le \period + \players \).
\end{proof}

\section{Learning players}
\label{sec:learning}

In this section we study a model where every player independently uses a no-regret strategy. 
The no-regret property, originally introduced by \citet{hannanApproximationBayesRisk1957}, is a property of multiple algorithms used in online reinforcement learning~\citep[see, e.g.,][]{perchetApproachabilityRegretCalibration2014}. 
It specifies that in hindsight, the actions taken by a player are at least (asymptotically) as good as any constant action. Formally, in the one-player case, given a sequence of cost functions $\rcost$ indexed by time $\ti$ and a sequence of actions $\tact$, player $\pl$'s regret  is defined as
\begin{equation}\label{eq:regret}
  \regret_{\ti}^{\pl} = \max_{\pureb[\pl]\in \actions}\sum_{\trun=1}^{\ti}\rcost[\trun](\pureb[\pl]) - \rcost[\trun](\tact[\trun]).
\end{equation}

A strategy satisfies the \emph{no-regret} property if $\regret^{\pl}_\ti = o(\ti)$. Well-known no-regret strategies include regret-matching~\cite{hartSimpleAdaptiveStrategies2013}, stochastic fictitious play~\cite{fudenbergTheoryLearningGames1998}, and the \acf{EWA}~\cite{littlestoneWeightedMajorityAlgorithm1994, cesa-bianchiPredictionLearningGames2006}, which  we study below.

\paragraph{Exponential Weight Algorithm} 
In the following, we use the \acfi{EWA}\acused{EWA}, which is known to have no-regret guarantees when the payoff is bounded. 
Unfortunately, boundedness is not satisfied here, as the number of jobs in the system could grow to infinity, resulting in an
unbounded penalty. 
For this reason, we need to study the efficiency of the algorithm more closely.

In the context of repeated games, weights \glsadd{weight} are classically defined as follows:
\begin{equation}
\label{eq:weight}
  \weight = \exp\left(\sum_{\trun=1}^{\ti-1} - \smness \cost[{\pureb[\pl], \tact[\trun][-\pl]}]\right),
\end{equation}
where $\smness$ is a positive constant.\glsadd{smoothness}
\cref{eq:weight} can be rewritten in a recursive fashion as
\begin{equation}
  \weight[\ti+1] = \weight \exp\left(-\smness \cost[\pureb[\pl], \tact[\ti][-\pl]]\right).
\end{equation}
Then, the \ac{EWA} strategy specifies that action $\pureb[\pl]$ is chosen at time $\ti$  with probability
\begin{equation}
\label{eq:noregretstrat}  
\tstrat(\pureb[\pl]) = \frac{\weight}{\sum_{\purep[\pl] \in \actions} \weight[\ti][\purep[\pl]]}.
\end{equation}

\paragraph{\Acf{MLEWA}}

\ac{EWA} is not designed to handle a changing environment. 
Here, the number of jobs held by every player changes with time. 
Therefore, there we need to specify how such information is used. 
We design a new protocol called \ac{MLEWA} where each player uses several copies of \ac{EWA}. 
This protocol is indexed by a parameter \(\jobsrun\), which we call a \emph{level}.
 When the number of jobs that player \(\pl\) owns reaches a new level for the first time, this player starts a new \ac{EWA} where the weights are initialized as a function of the past. 
 When the number of jobs of player \(\pl\) equals a level that has already been visited, this player follows the recommendation given by \ac{EWA} for this level and updates the weights following \ac{EWA}.
Notice that at any given period \(\ti\) different players may use an \ac{EWA} at different levels.

Let $\firsttime$\glsadd{firsttime} be the first time  player \(\pl\) has \(\jobsrun\) jobs, with the convention that \(\firsttime=+\infty\) if player \(\pl\) never has \(\jobsrun\) jobs. 
Player \(\pl\)'s weights \(\ww\)  are now parameterized by two parameters: the period \(\ti\) and the level \(\jobsrun\).
The algorithm at level \(\jobsrun\) is defined for all \(\ti \geq \firsttime\) by induction. 
We start by describing the induction step, which is given by\glsadd{mlweight} 
\begin{equation}
\label{eq:onealgo}
\ww[t+1][\pureb[\pl]]=
\begin{cases} 
\ww[\ti][\pureb[\pl]] \exp\left(-\sm \ticost[{\pureb[\pl], \tact[\ti][-\pl]}] \right) & \text{ if }\tent=\jobsrun,\\
 \ww[\ti] & \text{ otherwise.} 
\end{cases}
\end{equation}
\cref{eq:onealgo} implies that \(\ww\) is updated if and only
if \(\tent=\jobsrun\). 

We now describe the initialization. 
At $\firsttime$, we start a new \ac{EWA} protocol and define initial weights
\begin{equation*}
    \ww[\firsttime][] := \ww[\firsttime][][\tent[\firsttime-1]],
\end{equation*}
that is, weights of a newly encountered state are defined as equal to the weights of the previously visited level.

\Cref{algo:mlewa} summarizes all the steps of our procedure. Its variables are not indexed by $\ti$ as there are a fixed number of variables and the algorithm does not access the whole history. Instead, it computes the new values at each step and updates the corresponding variables.

\Ac{MLEWA} is based on a no-regret algorithm adapted to changing states. \cref{eq:regret} does not deal with changing states. This is a limitation well identified by \citeauthor{GaiTar:EC2020}, where regret is computed assuming the state path is unchanged by a change of action.
The situation becomes much more complicated when states change endogenously.
What we can say about our algorithm is that it has the no-regret property state-by-state.

We can now state our main result about the stability of the system when players learn.

\begin{algorithm}[t]
	\SetAlgoNoLine
    $\forall \pureb[\pl] \in \actions, \text{ initialize } \ww[][\pureb[\pl]][1]$ \;
    $\forall i \in \players, \ent \gets 1$  \tcp*{level $1$}
    \For{each step $t \geq 1$}{
        \For{each player $\pl$} {
            select an action $\tact \sim \tststrat[\pl][][\tent[]]$
            \tcp*{proportional to $\ww[][][\ent]$}
        }
        \For{each player $\pl$} {
            $\forall \pureb[\pl] \in \actions, \ww[][\pureb[\pl]][\ent] \gets \ww[][\pureb[\pl]][\ent] \exp(-\smness \tcost[\pureb[\pl], \tact[\ti][-\pl]])$ \tcp*{number of late jobs + 1}
            $\ent \gets \latecount +1$\;
            \If{$\ww[][][\ent]$ is not defined}{
                $\forall \pureb[\pl] \in \actions, \text{ initialize }\ww[][\pureb[\pl]][\ent]$ \tcp*{level $\ent$}
            }
        }
    }
 \caption{\acf{MLEWA}}
 \label{algo:mlewa}
\end{algorithm}

\begin{theorem}[Stability of joint no-regret strategies in the subcritical case]
  \label{thm:onealgolearning}
  If \(\players < \period\) and \(\tplty[\ent[]] > 4\nent \period\) for all \(\tnent\), then
  strategy profiles where all players use \ac{MLEWA} 
  are
  stable.
\end{theorem}
  
Several lemmas are needed to prove \cref{thm:onealgolearning}. 
First, we show that (in a static context) action $0$ strictly dominates any other actions for a player who holds a large enough number of jobs.
The implication of this dominance in our dynamic model is that that the weight on action $0$ converges towards $1$ when enough jobs are held by a player. 
Finally, we expose some results on reinforced random walks.

\subsection{Domination by action \texorpdfstring{$0$}{0}}
The following lemma shows that in the static case action $0$ is strictly dominant  for a player $\pl$ who has enough jobs. 

\begin{lemma}
[Strict Domination by \(0\) when \(\tent > 2\period^2\) in the static model]
  \label{lem:domlargek}
  If \(\tent > 2\period^{2}\),
  \(\tplty > 4 \tnent \period\), and \(\pureprof\) is an action profile such that $\purep[\pl] \neq 0$, then
  \(\tcost[0, \pureprof[-\pl]] < \tcost[\pureprof] - \tent\).
\end{lemma}

\begin{proof}
  
  Let \(\pureprof\) be a pure action profile such that $\purep[\pl] \neq 0$. 
  Call
  \begin{equation}
  \label{eq:k-zero}  
    \tkz = \sum_{\plb \neq \pl} \tent[\plb] 1_{\purep[\plb] = 0} 
  \end{equation}
  the number of jobs that join the queue at \(0\) excluding player's \(\pl\) jobs.
 
  Then:
  \begin{equation*}
    \costz = \tent \braces*{\period + \frac{\tent + \tkz - \period}{\tent + \tkz}\tplty} 
    = \tent\braces*{\period+ \braces*{1-\frac{\period}{\tent+\tkz}} \tplty}.
  \end{equation*}
  
  For each job, \(\pl\) incurs the waiting cost \(\period\) and an additional cost due to the probability of being late. 
  At period $0$, \(\tent+\tkz > \period\)
  jobs join the queue; therefore some jobs will surely be late. 
  Moreover, the probability that a job does not incur the penalty is equal to the probability that this job joins the queue among the first \(\period\) jobs, which happens with probability \(\period/(\tent+\tkz)\).
 
  \begin{itemize}
  \item
    If \(\tkz \geq \period\), then under the profile $\pureprof$ the queue is full from stage $0$ and the jobs sent by $\pl$ are all late, \ie \(\costb = \tent(\period - \myb) + \tent \tplty\). 
    Then
    \begin{equation*}
        \costz - \costb = -\frac{\tent \period}{\tent+\tkz}\tplty + \tent \myb.        
    \end{equation*}
    The assumption on \(\tplty\) implies that
    \(\tplty > \tnent (\period+1)/\period\), so:
    \begin{equation*}
        \costz - \costb \leq - \frac{\tent}{\tent + \tkz} \tnent (\period+1) + \tent \myb.
    \end{equation*}
    Since, by definition, \(\tnent \geq \tent + \tkz\), it follows that:
    \[
        \costz - \costb \leq - \tent (\period+1) + \tent \myb \leq \tent (\myb - \period) -\tent \leq -\tent.
    \]
  \item
    Consider now the case \(\kz < \period\). 
    Player \(\pl\) pays the waiting cost \(\period-\myb\) for each job; 
    at most \(\period - \myb\) of player \(\pl\)'s jobs can leave the system without being late.
    Consequently the following bound holds:
    \[
        \costb \geq \tent (\period - \myb) + (\tent - \period + \myb)\tplty.
    \]
    Then
    \begin{equation}
      \begin{aligned}
        \costz - \costb 
        & = \tent\period+ \tent\tplty-\frac{\period\tent}{\tent+\tkz}\tplty \\
        &\quad-\tent \period+\tent\myb - (\tent - \period + \myb)\tplty,\\
        & \leq -\frac{\tent \period}{\tent + \tkz} \tplty + \tent \myb + \period \tplty - \myb \tplty,\\
        & = \parens*{\frac{\tkz \period}{\tent+\tkz} - \myb}\tplty + \tent \myb.
      \end{aligned}
      \label{eq:k0T}
    \end{equation}
    Since \(\tent > 2\period^2\), it follows that \(\tent + \tkz \geq 2\period^2\), so
    \[
    \frac{\tkz \period}{\tent + \tkz} \leq \frac{\period^{2}}{2\period^{2}} < \frac{1}{2}.
    \]
    Hence,
    \[
    \costz - \costb \leq \parens*{\frac{1}{2}-\myb}\tplty + \ent \myb < - \frac{\tplty}{2} + \tent \period,
    \]
    because \(1\leq \myb\leq \period\).

    The assumption that \(\tplty > 4 \nent \period\) implies
    \[
     \costz - \costb <  - 2 \tnent \period + \tent \period < -\tnent \period<-\tnent<-\tent.
     \qedhere
     \]
  \end{itemize}
\end{proof}

\subsection{Action \texorpdfstring{$0$}{0} is increasingly preferred when the number of jobs grows}

We now use \cref{lem:domlargek} to show that, when the number of jobs in the system is high enough, for strategies \(\vx\)  defined as in \cref{eq:noregretstrat},  the weight \(\ww[\ti][0]\) of action \(0\) increases faster than the other weights. 
Since  $\vx$ is  proportional to this weight,  players are more and more prone to play action \(0\) when the state is visited again.

Call $\countn$\glsadd{countn} the number of times that the level of player $\pl$ is $\countvar$, up to time $\ti-1$:
\begin{equation}
  \countn = \#\braces*{\tent[\trun] = \countvar \mid \trun \in \braces*{0, \dots, \ti-1}}.
\end{equation}
We have
\begin{lemma}[Preference for \(0\)]\label{lem:pref0gen}
  If \(\countvar > 2\period ^2\) and
  \(\tplty > 4\nent \period\), then for all \(\ti \geq \firsttime\),
\begin{equation*}
\vx(0) \geq \frac{\tststrat[\pl][\firsttime](0)}{\tststrat[\pl][\firsttime](0)+\left(1-\tststrat[\pl][\firsttime](0)\right) \exp\left(-\smness\countvar\countn\right)}.
\end{equation*}
\end{lemma}

\begin{proof}

  We first prove by induction on $\ti$ that for every \(\pureb[\pl] \neq 0\), we get:
  \begin{equation}
  \label{pr:eq7}
    \frac{\ww}{\ww[\firsttime]} \leq \exp\left(-\sm\countvar\countn\right) \frac{\ww[\ti][0]}{\ww[\firsttime][0]},
  \end{equation}
  
  If $\ti = \firsttime$ then by definition player $\pl$ never had $\countvar$ jobs before and $\countn = 0$.

  It follows that both sides are equal to $1$ and the result is true.
  
  \newcommand{\countm}{\countn[\countvar][\ti-1]}
  
  We now show that, if the result holds for \(\ti-1\), then it holds for \(\ti\). 
  There are two cases.

  If $\tent[\ti-1] \neq \countvar$, then all weights are equal at stage $\ti$ and $\ti-1$, i.e. $\ww[\ti]=\ww[\ti-1]$ for all action $\pureb[\pl]$. Moreover $\countn[\countvar][\ti]=\countm$, so the inequality is the same at $\ti$ and $\ti-1$ and therefore true. If $\tent[\ti-1]= \countvar$, then $ \countn=\countm+1$.
  By the recurrence hypothesis, we know that
  \begin{equation}
  \label{pr:eq8}
    \frac{\ww[\ti-1]}{\ww[\firsttime]} \leq \exp\left(-\sm\countvar\countm\right) \frac{\ww[\ti-1][0]}{\ww[\firsttime][0]}.
  \end{equation}
  Using \cref{lem:domlargek,eq:onealgo}, for \(\pureb[\pl] \neq 0\), we get
  \begin{align}
  \label{pr:eq9}
    \frac{\ww}{\ww[\ti-1]} & =\exp\left(-\sm \ticost[{\pureb[\pl], \tact[\ti-1][-\pl]}][\pl][{\vent}] \right) \\
    & \leq \exp\left(-\sm\countvar\right)\exp\left(-\sm \ticost[{0, \tact[\ti-1][-\pl]}][\pl][{\vent}] \right),\\
    & = \exp(-\sm\countvar ) \frac{\ww[\ti][0]}{\ww[\ti-1][0]},   \label{pr:eq10}
  \end{align}
  Multiplying  \cref{pr:eq8,pr:eq10}, we obtain the result for \(\ti\).
 
  Fix now \(\ti\). \cref{pr:eq7} implies that
  \begin{equation*}
    \ww \leq \exp\left(-\sm\countvar\countn\right) \frac{\ww[\ti][0]}{\ww[\firsttime][0]}\ww[\firsttime].
  \end{equation*}

  Therefore,

  \begin{equation*}
    \begin{aligned}
    \tststrat(0) & = \frac{\ww[\ti][0]}{\sum_{\pureb[\pl] \in \actions} \ww[\ti][\pureb[\pl]] }
    \\ & = \frac{\ww[\ti][0]}{\ww[\ti][0]+\sum_{\pureb[\pl]\neq 0} \ww[\ti]}
    \\ & \geq \frac{\ww[\ti][0]}{\ww[\ti][0]+\sum_{\pureb[\pl]\neq 0}  \exp\left(-\sm\countvar\countn\right) \frac{\ww[\ti][0]}{\ww[\firsttime][0]}\ww[\firsttime]}
    \\ & = \frac{\ww[\firsttime][0]}{\ww[\firsttime][0]+\sum_{\pureb[\pl]\neq 0}  \exp\left(-\sm\countvar\countn\right)\ww[\firsttime]} 
    \\ & = \frac{\tststrat[\pl][\firsttime](0)}{\tststrat[\pl][\firsttime](0)+\sum_{\pureb[\pl]\neq 0}  \exp\left(-\sm\countvar\countn\right)\tststrat[\pl][\firsttime](\pureb[\pl])} 
    \\ & =\frac{\tststrat[\pl][\firsttime](0)}{\tststrat[\pl][\firsttime](0)+\left(1-\tststrat[\pl][\firsttime](0)\right) \exp\left(-\smness\countvar\countn\right)},
    \end{aligned}
  \end{equation*}
  which proves the lemma.
\end{proof}

\begin{lemma}\label{lem:initbound}
There exists $B^\pl > 0$ such that \begin{equation}\initializationnormalized(0)\geq \frac{1}{1+B^\pl\exp(-\smness\countvar)}.\end{equation}
\end{lemma}
 \begin{proof} 
  Let \(\pl\in \players\). We can define
  \[
  Z^\pl \coloneqq \max_{\countvar \leq 2\period^{2}\text{ s.t. } \firsttime <+\infty} \frac{1}{\initializationnormalized(0)}-1 .
  \]
  
  For $\countvar=1$, we initialize the algorithm uniformly so every action has initially a strictly positive weight. 
  By definition of \ac{EWA}, if an action has a strictly positive weight during the initialization then it is always played with strictly positive probability. 
  When reaching the level $\countvar=2$, the initialization is done by copying the current distribution of the algorithm of level $\countvar=1$, hence each action has  a strictly positive weight too. 
  Induction proves that at every stage and for every level, the probability to play every action is strictly positive---and  strictly lower than \(1\).
  It follows that $Z^\pl$ is strictly positive  as the minimum of finitely many strictly positive numbers.

  By definition, for every $\jobsrun \leq 2 \period^{2}$ such that $\firsttime <+\infty$, one has 
  \[
  \initializationnormalized \ge \frac{1}{1+Z^\pl} \ge \frac{1}{1+Z^\pl\exp(\smness 2\period^2)\exp(-\smness \countvar)},
  \]
  so let $B^\pl \coloneqq Z^\pl \exp(\smness 2\period^2)$.

  We now prove that this is true also for $\jobsrun > 2 \period^{2}$. 
  The proof is by induction. 
  Assume that it is true for $\jobsrun\geq 2 \period^{2}$ and consider $\countvar+1$ such that $\firsttime[\pl][\countvar+1]<+\infty$. 
  Since the  increment in the number of jobs is at most one, this implies that $\firsttime[\pl][\countvar]<+\infty$.
  
  By \cref{lem:pref0gen}, for all \(\ti > \firsttime\), the weight on $0$ satisfies
    \begin{align}
    \tststrat(0) & \geq \frac{\tststrat[\pl][\firsttime](0)}{\tststrat[\pl][\firsttime](0)+\left(1-\tststrat[\pl][\firsttime](0)\right) \exp\left(-\smness\countvar\countn\right)} \notag \\
    & = \frac{1}{1+\left(\frac{1}{\tststrat[\pl][\firsttime](0)}-1\right) \exp\left(-\smness\countvar\countn\right)} \notag \\
    & \geq \frac{1}{1+\left(1+B^\pl\exp(-\smness\countvar)-1\right) \exp\left(-\smness\countvar\countn\right)} \label{eq:hr1}\\
    & \geq \frac{1}{1+B^\pl\exp\left(-\smness\countvar(1+\countn)\right)} \notag \\
    & \geq \frac{1}{1+B^\pl\exp\left(-\smness(\countvar+1)\right)} \label{eq:hr2}
    \end{align}
    where \cref{eq:hr1} follows from the recurrence hypothesis and \cref{eq:hr2} is implied by $t > \firsttime$, so $\countn \geq 1$.
    
    The initial weight when reaching $\countvar+1$ for the first time is equal to the current weight for $\countvar$ packages, it follows that
    \begin{equation}
     \tststrat[\pl][\firsttime[\pl][\countvar+1]][\countvar+1](0)=     \tststrat[\pl][\firsttime[\pl][\countvar+1]][\countvar](0) \geq  \frac{1}{1+B^\pl\exp(-\smness(\countvar+1))}.
    \end{equation} 

This proves the result for $\jobsrun+1$. Hence, it concludes the induction and proves the lemma.
\end{proof}

\subsection{Reinforced Random Walks}

\cref{lem:pref0gen} shows that every time the process reaches a given level, there is a reinforcement on the probability to play the action profile where every player plays $0$. 
The next step is to understand how this reinforcement influences the system dynamic. 
In order to do so, we prove some results on reinforced random walks. 
We follow the presentation of
\citep[p.~47]{menshikovNonhomogeneousRandomWalks2017} of nearest neighbor one-dimension random walk. 
They study  random walks that are non-homogeneous \emph{in space} but homogeneous \emph{in time}. 
The difference is that we suppose there is a reinforcement factor in the drift, resulting in a random walk that is  non-homogeneous \emph{in time and space}, but
bounded. 
Furthermore, we suppose that there is more heterogeneity in the weight of our
random walk, in the sense that precise probabilities of going up or down are highly dependent of the past but nevertheless bounded.

The proof of \cref{thm:onealgolearning} requires the following lemma, whose proof can be found in \cref{se:proofs}.

\begin{lemma}[Reinforced Random Walk]\label{lem:reinfwalk3-main}
    Let \(\maxmove > 0\) and \((\walk, Z_\ti, \ti \geq 0)\) a sequence of random variables in
    $\N$ such that $\walk/d \leq Z_\ti \leq \walk$ with $d > 1$, \(|\walk[\nti] - \walk| \leq \maxmove\) and
    \(\sigwalk = \sigma(\walk[0], Z_0, \ldots, Z_{\ti}, \walk)\). Suppose that there exists
    a function $r: \N^2 \rightarrow {\mathbb R}^+$, reals $\levz$ and $\sumpr > 0$ such that:
    
    \begin{itemize}
    \item
      for all $\ti \ge 0$ and $Z_t \geq \levz$, \(\proba[\walk[\nti] > \walk][\sigwalk] \leq \princr[Z_t][\countn[Z_t]]\) almost surely,
      where \(\countn[z]\) is the number of occurrences of the \(Z_\ti=z\)
      event for \(\tib \leq \ti\),
    \item
      for \(z \geq \levz\),  \(\sum_{\occl} \princr[z] < \frac{\sumpr}{z}\)

    \end{itemize}
    
    Then \(\walk\) is almost surely bounded.
\end{lemma}

\begin{proof}[Proof of \cref{thm:onealgolearning}]

  Let $\walk := \players \tnent$ and $Z_t := \max_{\plb \in \players}\left(\players\tent[\ti][\plb]+\plb\right)$ be random variables that are an encoding of $\tnent$, $\tent[\ti][\plb]$, $\plb$ where $\plb$ maximizes $\tent[\ti][\plb]$ and is maximal among the maximizers. Indeed, $\tnent = \walk/\players$, $\tent[\ti][\plb] = \left[Z_t/\players\right]$ and $\plb = Z_t \mod \players$.

  By definition,
  \begin{equation}
      Z_t = \max_{\plb\in \players}\left(\players\tent[\ti][\plb]+\plb\right) \leq \max_{\plb\in \players} \left(\players(\tnent-1)+\plb\right) \leq \players \tnent = \walk.
  \end{equation}
  Moreover,
  \begin{equation}
      Z_t \ge \max_{\plb \in \players} \players \tent[\ti][\plb] \ge \players \max_{\plb \in \players} \tent[\ti][\plb] \ge \players \frac{\tnent}{\players} = \frac{\walk}{\players}.
  \end{equation}

  Let $\levz = 2\players \period^2$. Suppose $Z_t \geq \levz$ and let $\plb$ such that $Z_t = \players\tent[\ti][\plb]+\plb$. In the following, we write $\ent[\plb]$ for $\tent[\ti][\plb]$ and $\firsttime[\plb][]$ for $\firsttime[\plb][\ent[\plb]]$. Then the probability that $\plb$ plays $0$ is $\tststrat[\plb][\ti][\ent[\plb]](0)$ which by \cref{lem:pref0gen} satisfies
  \begin{align}
    \tststrat[\plb][\ti][\ent[\plb]](0)  & \geq \frac{\tststrat[\plb][\firsttime[\plb][]][\ent[\plb]](0)}{\tststrat[\plb][\firsttime[\plb][]][\ent[\plb]](0)+\left(1-\tststrat[\plb][\firsttime[\plb][]][\ent[\plb]](0)\right) \exp\left(-\smness\ent[\plb] \countn[\ti][\ent[\plb]]\right)} \\
     & = \frac{1}{1+\left(\frac{1}{\tststrat[\plb][\firsttime[\plb][]][\ent[\plb]](0)}-1\right) \exp\left(-\smness\ent[\plb] \countn[\ti][\ent[\plb]]\right)}.
    \label{eq:pr1}
  \end{align}
  Moreover, by \cref{lem:initbound}, there exists $B>0$ such that
  \(
  \tststrat[\plb][\firsttime[\plb][]][\ent[\plb]](0) \geq \frac{1}{1+B\exp(-\smness \ent[\plb])},
  \)
  hence from \cref{eq:pr1}
  \begin{equation}
    \tststrat[\plb][\ti][\ent[\plb]](0) \geq \frac{1}{1+B\exp(-\smness\ent[\plb])\exp(-\smness \ent[\plb]\countn[\ti][\ent[\plb]])}.
    \label{eq:pr2}
  \end{equation}
  
  At each period, there are $\players$ new jobs. By \cref{hyp:rep}, there are less than $\period$ new jobs. Since $\plb$ has more than $\left[\levz/\players\right] = 2\period^2$ jobs, when $\plb$ plays $0$, we know that at least $\period$ jobs are not late. Therefore, the number of jobs at the next period has to be smaller or equal compared to the current period. Hence,
  \begin{equation}\proba[X_{\ti+1} > X_{\ti}][\sigwalk] \leq 1-\tststrat[\plb][\ti][\ent[\plb]](0) \leq \frac{B\exp(-\smness\ent[\plb])\exp(-\smness \ent[\plb]\countn[\ti][\ent[\plb]])}{1+B\exp(-\smness\ent[\plb])\exp(-\smness \ent[\plb]\countn[\ti][\ent[\plb]])}
  \label{eq:pr8}
  \end{equation}
  using \cref{eq:pr2}.

  This suggests the following definition,
  \begin{equation}
  r(Z_t, \countn[\ti][Z_t]) := B\exp(-\smness[Z_t/\players])\exp(-\smness [Z_t/\players]\countn[\ti][Z_t])
  \label{eq:pr4}
  \end{equation}
  which is equal to
  \[
  B\exp(-\smness\ent[\plb])\exp(-\smness \ent[\plb]\countn[\ti][Z_t]),
  \]
  because $\plb = Z_t \mod \players$, $\ent[\plb] = \left[Z_t/\players\right]$ and consequently, $\countn[\ti][Z_t]$ (the number of times $Z_t$ was equal to the current value) is lower than $\countn[\ti][\ent[\plb]]$ (the number of times that $\plb$ had the current number of jobs). Therefore, 
  \begin{equation}
  \label{eq:pr3}
    \exp(-\smness \ent[\plb]\countn[\ti][\ent[\plb]]) \leq \exp(-\smness \ent[\plb] \countn[\ti][Z_t]).
  \end{equation}
  
  Using previous equations,
  \begin{align}
    \proba[X_{t+1} > X_t][\sigwalk] & \leq B\exp(-\smness\ent[\plb])\exp(-\smness \ent[\plb]\countn[\ti][\ent[\plb]]) \label{eq:pr5} \\ 
    & \le B\exp(-\smness\ent[\plb])\exp(-\smness \ent[\plb]\countn[\ti][Z_t]) \label{eq:pr6} \\
    & \leq r(Z_t, \countn[\ti][Z_t]) \label{eq:pr7}
  \end{align}
  where \ref{eq:pr5} comes from \cref{eq:pr8}, \ref{eq:pr6} from \cref{eq:pr3} and \ref{eq:pr7} from \cref{eq:pr4}.
  
  For all $z \geq z_0$, the sum on $m$ of $r(z, m)$ is
  \[
    B\exp(-\smness[z/\players])\frac{1}{1-\exp(-\smness[z/\players])},
  \]
  so it is bounded by $A/z$ for some $A > 0$ and \cref{lem:reinfwalk3-main} applies, so we proved the theorem.
\end{proof}

\section{Conclusions}

We have studied a repeated strategic queueing model with spillover from one period to another. 
We have focused on the stability of the system when players play learning strategies. 
Several problems remain open in this model.

\paragraph{Multi-Level regret} 
We have used a \acl{MLEWA}.
Although \ac{MLEWA} is based on a no-regret algorithm, to prove that it is \emph{itself} a no-regret algorithm, we would need an appropriate definition of regret in the context of endogenously changing states. 
Proving the no-regret property of \ac{MLEWA} with a suitable definition of regret and using it to prove the system stability would be an interesting generalization.
Another promising research direction is the definition of other multi-level algorithms based on different no-regret algorithms.
In particular, it would be important to see which stability properties depend on the specific algorithm used and which other properties are general and hold for every no-regret algorithm. 

\paragraph{Model}
\label{pa:future-work}
Several extensions of the model are conceivable. 
For instance, a model with more than one server could be studied. 
In that case the strategy of each player would have two components: the chosen server and the chosen time at which jobs join the chosen server's queue.
An apparently simple, but non-trivial generalization would involve the consideration of lower penalty costs.

\paragraph{Importance of the value of $\tplty[\ent[]]$} 
Several results of our paper are based on the value of the penalty $\tplty[\ent[]]$. 
If it is large enough, then the system is stable. 
We conjecture that a large but constant penalty is not sufficient for the stability in the learning context; \ie the penalty must depend on the number of jobs in the system, otherwise the number of jobs could be unbounded with a positive probability. 
This contrasts with the myopic strategic case, where a constant penalty cost guarantees stability, if it is large enough.

\subsection*{Acknowledgments}

Marco Scarsini and Xavier Venel are members of GNAMPA-INdAM.
Part of this work was carried out when Lucas Baudin was visiting Luiss University. 
This research project received partial support from  the Italian MIUR PRIN 2017 Project ALGADIMAR ``Algorithms, Games, and Digital Markets.''
Lucas Baudin ackowleddes the financial support by COST action GAMENET CA16228. 
Xavier Venel acknowledges the financial support by the National Agency for Research, Project CIGNE (ANR-15-CE38-0007-01).

\bibliographystyle{apalike}
\bibliography{bibliography}

\begin{thebibliography}{}

\bibitem[Anagnostides et~al., 2023]{AnaPanFarSan:arXiv2023}
Anagnostides, I., Panageas, I., Farina, G., and Sandholm, T. (2023).
\newblock On the convergence of no-regret learning dynamics in time-varying
  games.
\newblock Technical report, arXiv:2301.11241.

\bibitem[Besbes et~al., 2015]{BesGurZee:OR2015}
Besbes, O., Gur, Y., and Zeevi, A. (2015).
\newblock Non-stationary stochastic optimization.
\newblock {\em Oper. Res.}, 63(5):1227--1244.

\bibitem[Besbes et~al., 2019]{BesGurZee:SS2019}
Besbes, O., Gur, Y., and Zeevi, A. (2019).
\newblock Optimal exploration-exploitation in a multi-armed bandit problem with
  non-stationary rewards.
\newblock {\em Stoch. Syst.}, 9(4):319--337.

\bibitem[Cardoso et~al., 2019]{CarAbeWanXu:ICML2019}
Cardoso, A.~R., Abernethy, J., Wang, H., and Xu, H. (2019).
\newblock Competing against {N}ash equilibria in adversarially changing
  zero-sum games.
\newblock In Chaudhuri, K. and Salakhutdinov, R., editors, {\em Proceedings of
  the 36th International Conference on Machine Learning}, volume~97 of {\em
  Proceedings of Machine Learning Research}, pages 921--930. PMLR.

\bibitem[Cesa-Bianchi and Lugosi,
  2006]{cesa-bianchiPredictionLearningGames2006}
Cesa-Bianchi, N. and Lugosi, G. (2006).
\newblock {\em Prediction, Learning, and Games}.
\newblock Cambridge University Press, Cambridge, UK.

\bibitem[Crippa et~al., 2022]{CriGurLig:arXiv2022}
Crippa, L., Gur, Y., and Light, B. (2022).
\newblock Regret minimization with dynamic benchmarks in repeated games.
\newblock Technical report, arXiv:2212.03152.

\bibitem[Duvocelle et~al., 2022]{DuvMerStaVer:MOR2022}
Duvocelle, B., Mertikopoulos, P., Staudigl, M., and Vermeulen, D. (2022).
\newblock Multiagent online learning in time-varying games.
\newblock {\em Math. Oper. Res.}, forthcoming.

\bibitem[Fudenberg and Levine, 1998]{fudenbergTheoryLearningGames1998}
Fudenberg, D. and Levine, D.~K. (1998).
\newblock {\em The Theory of Learning in Games}.
\newblock MIT Press, Cambridge, MA.

\bibitem[Gaitonde and Tardos, 2020]{GaiTar:EC2020}
Gaitonde, J. and Tardos, E. (2020).
\newblock Stability and learning in strategic queuing systems.
\newblock In {\em Proceedings of the 21st ACM Conference on Economics and
  Computation}, EC '20, page 319–347, New York, NY, USA. Association for
  Computing Machinery.

\bibitem[Gaitonde and Tardos, 2021]{GaiTar:EC2021}
Gaitonde, J. and Tardos, E. (2021).
\newblock Virtues of patience in strategic queuing systems.
\newblock In {\em Proceedings of the 22nd ACM Conference on Economics and
  Computation}, EC '21, page 520–540, New York, NY, USA. Association for
  Computing Machinery.

\bibitem[Glazer and Hassin, 1983]{GlaHas:EJOR1983}
Glazer, A. and Hassin, R. (1983).
\newblock {$?/M/1$}: on the equilibrium distribution of customer arrivals.
\newblock {\em European J. Oper. Res.}, 13(2):146--150.

\bibitem[Hannan, 1957]{hannanApproximationBayesRisk1957}
Hannan, J. (1957).
\newblock Approximation to {B}ayes risk in repeated play.
\newblock In {\em Contributions to the Theory of Games, Vol. 3}, Annals of
  Mathematics Studies, no. 39, pages 97--139. Princeton University Press,
  Princeton, N.J.

\bibitem[Hart and Mas-Colell, 2013]{hartSimpleAdaptiveStrategies2013}
Hart, S. and Mas-Colell, A. (2013).
\newblock {\em Simple Adaptive Strategies. From Regret-Matching to Uncoupled
  Dynamics}.
\newblock World Scientific Publishing Co. Pte. Ltd., Hackensack, NJ.

\bibitem[Hassin, 1985]{Has:E1985}
Hassin, R. (1985).
\newblock On the optimality of first come last served queues.
\newblock {\em Econometrica}, 53(1):201--202.

\bibitem[Hassin, 2016]{Has:CRC2016}
Hassin, R. (2016).
\newblock {\em Rational Queueing}.
\newblock CRC Press, Boca Raton, FL.

\bibitem[Hassin and Haviv, 2003]{HasHav:Kluwer2003}
Hassin, R. and Haviv, M. (2003).
\newblock {\em To Queue or Not to Queue: Equilibrium Behavior in Queueing
  Systems}.
\newblock Kluwer Academic Publishers, Boston, MA.

\bibitem[Haviv and Ravner, 2021]{Haviv_Ravner:2021}
Haviv, M. and Ravner, L. (2021).
\newblock A survey of queueing systems with strategic timing of arrivals.
\newblock {\em Queueing Syst.}, 99(1-2):163--198.

\bibitem[Kawasaki et~al., 2023]{KawKonYuk:IJGT2023}
Kawasaki, R., Konishi, H., and Yukawa, J. (2023).
\newblock Equilibria in bottleneck games.
\newblock {\em Internat. J. Game Theory}, forthcoming.

\bibitem[Littlestone and Warmuth,
  1994]{littlestoneWeightedMajorityAlgorithm1994}
Littlestone, N. and Warmuth, M.~K. (1994).
\newblock The weighted majority algorithm.
\newblock {\em Inform. and Comput.}, 108(2):212--261.

\bibitem[Menshikov et~al., 2017]{menshikovNonhomogeneousRandomWalks2017}
Menshikov, M., Popov, S., and Wade, A. (2017).
\newblock {\em Non-Homogeneous Random Walks}.
\newblock Cambridge University Press, Cambridge, UK.

\bibitem[Mertikopoulos and Staudigl, 2021]{MerSta:IEEECDC2021}
Mertikopoulos, P. and Staudigl, M. (2021).
\newblock Equilibrium tracking and convergence in dynamic games.
\newblock In {\em 2021 60th IEEE Conference on Decision and Control (CDC)},
  pages 930--935.

\bibitem[Naor, 1969]{Naor:1969}
Naor, P. (1969).
\newblock The regulation of queue size by levying tolls.
\newblock {\em Econometrica}, 37(1):15--24.

\bibitem[Pemantle, 2007]{pemantleSurveyRandomProcesses2007}
Pemantle, R. (2007).
\newblock A survey of random processes with reinforcement.
\newblock {\em Probab. Surv.}, 4:1--79.

\bibitem[Perchet, 2014]{perchetApproachabilityRegretCalibration2014}
Perchet, V. (2014).
\newblock Approachability, regret and calibration: implications and
  equivalences.
\newblock {\em J. Dyn. Games}, 1(2):181--254.

\bibitem[Rivera et~al., 2018]{RivScaToM:SSRN2018}
Rivera, T., Scarsini, M., and Tomala, T. (2018).
\newblock Efficiency of correlation in a bottleneck game.
\newblock Technical report, SSRN3219767.

\bibitem[Sentenac et~al., 2021]{sentenac2021decentralized}
Sentenac, F., Boursier, E., and Perchet, V. (2021).
\newblock Decentralized learning in online queuing systems.
\newblock {\em Advances in Neural Information Processing Systems},
  34:18501--18512.

\bibitem[Vickrey, 1969]{Vic:AER1969}
Vickrey, W.~S. (1969).
\newblock Congestion theory and transport investment.
\newblock {\em Amer. Econ. Rev.}, 59(2):251--260.

\bibitem[Zhang et~al., 2022]{ZhaZhaLuoZho:ICML2022}
Zhang, M., Zhao, P., Luo, H., and Zhou, Z.-H. (2022).
\newblock No-regret learning in time-varying zero-sum games.
\newblock In Chaudhuri, K., Jegelka, S., Song, L., Szepesvari, C., Niu, G., and
  Sabato, S., editors, {\em Proceedings of the 39th International Conference on
  Machine Learning}, volume 162 of {\em Proceedings of Machine Learning
  Research}, pages 26772--26808. PMLR.

\end{thebibliography}

\appendix

\section{Symbols and acronyms}
\label{se:symbols}
\printglossary[title=Symbols]

\section*{Acronyms}
\begin{acronym}
\acro{FIFO}{first-in first-out}
\acro{LIFO}{last-in first-out}
\acro{CCE}{coarse correlated equilibrium}
\acroplural{CCE}[CCE]{coarse correlated equilibria}
\acro{DQM}{dynamic queueing model}
\acro{EWA}{exponential weight algorithm}
\acro{MLEWA}{multi-level EWA}
\acro{NE}{Nash equilibrium}
\end{acronym}

\section{Proof of the reinforced random walk lemma}
\label{se:proofs}

\NewDocumentCommand{\debugg}{m}{#1}
\NewDocumentCommand{\levb}{}{z}
\NewDocumentCommand{\di}{}{d}

\begin{proof}[Proof of \cref{lem:reinfwalk3-main}]

    \NewDocumentCommand{\tiz}{}{t_0}
    \NewDocumentCommand{\maxeq}{}{T}
    \NewDocumentCommand{\eqlev}{O{\ti}}{S^{}(#1)}
    \NewDocumentCommand{\reached}{O{\lev}}{A(#1)}
    \NewDocumentCommand{\intreached}{O{}}{\overline{R}^{#1}(\lev)}
    \NewDocumentCommand{\tireached}{}{T}
    \NewDocumentCommand{\prtime}{O{\ti}}{R_{#1}}
    \NewDocumentCommand{\probasup}{}{\proba[\sup_{\ti \in \N} \walk \in [\lev -\maxmove, \lev]][\reached] }
    \RenewDocumentCommand{\lev}{}{\debugg{x}}
    \RenewDocumentCommand{\levz}{}{\debugg{x_0}}

    The following proof is inspired by \citet{pemantleSurveyRandomProcesses2007}.
   
    We first prove that the probability that
    \(\sup_{\ti \in \N} \walk \in [\lev-\maxmove, \lev]\) for all $\lev$ such that $\lev \geq \levz + \maxmove$, conditionally on the fact that \([\lev-\maxmove, \lev]\) is reached, is lower bounded by something strictly greater than $0$ and that does not
    depend on \(\lev\). 
    As we will show, this implies that almost surely there exists  $\lev$ such that $\sup_{\ti \in \N} \walk = \lev$.
    
    Denote the event that $[\lev-\maxmove, \lev]$ is reached by $\walk$ by \(\reached\). 
    We now show that for all $\lev \geq \levz+\maxmove$:
    \begin{equation}\label{eq:app:supwalk}
        \probasup \geq \prod_{\levb \in \debugg{\left[\frac{\lev - \maxmove}{\di}, \lev\right]}} \prod_{m \in \N} (1-\princr[\levb])
    \end{equation}
To prove this result, we fix $\lev\in \N$ such that $\lev > d z_0+\maxmove$ and we introduce an extended random process. 
Define the new state space $\overline{\Omega}=\N^2 \times \N \times \N^{x} \times \{0,1\}$. 
The interpretation of the a state $(x,z,n_0,...,n_{x},i)$ is the following:
    \begin{itemize}
    \item the current state is $(x, z)$,
    \item the path of $Z_\ti$ has gone $n_r$ times through the state $r$,
    \item $i$ is equal to $1$ if and only if $\walk$ went up from a state between $x$ and $x-M$.
    \end{itemize}
    Formally, let $(X,Y,N_0,...,N_{x},I)_{t\geq 1}$ be the random process on $\overline{\Omega}$. The first coordinate is equal to $\walk$ whereas all other coordinates are deduced from it. By construction, we know that $\walk$ has a maximal increment of $M$, hence in order for the supremum to be strictly greater than $x$, it is necessary for a positive jump from a state between $x-M$ and $x$, hence
    \[
    \{\sup_{\ti \in \N} \walk>x\} \subset\{\exists t\geq 1, I_t=0\}.
    \]
    It follows that
    \[
    \probasup \geq \proba[\forall t\geq 1, I_t=0][\reached].
    \]
    Moreover, 
    \[
    \proba[\forall t\geq 1, I_t=0][\reached]\geq  \prod_{\levb \in [\frac{\lev - \maxmove}{\di}, \lev]} \prod_{m \in \N} (1-\princr[z]).
    \]
    Indeed, by construction of the auxiliary random process, we know that:
    \begin{itemize}
    \item for every $r\in \{0,...,M\}$, $N_r$ is only increasing,
    \item conditionally on $Z_t=z \geq \frac{x-M}{d}$, $N_z=n$ and the past, the probability for $i$ to stay equal to $0$ is at least $(1-\princr[z][n])$,
    \item conditionally on $z<\frac{x-M}{d}$,  the probability for $i$ to stay equal to $0$ is $1$.
    \end{itemize}
    It follows that 
    \begin{align}
        \probasup &\geq \proba[\forall t\geq 1, I_t=0][\reached]\\
        &\geq  \prod_{\levb \in [\debugg{\frac{\lev - \maxmove}{\di}}, \lev]} \prod_{m \in \N} (1-\princr[\levb]).
    \end{align}

    Then, the logarithm of the right hand side is
    \begin{equation*}
        \begin{aligned}
        \sum_{\levb \in [\debugg{\frac{\lev - \maxmove}{\di}}, \lev]}\sum_{m \in \N} \log(1-\princr[\levb]) & = \sum_{\levb \in [\debugg{\frac{\lev - \maxmove}{\di}}, \lev]}\sum_{m \in \N} - \log(1+\frac{\princr[\levb]}{1-\princr[\levb]}) \\
        & \geq \sum_{\levb \in [\debugg{\frac{\lev - \maxmove}{\di}}, \lev]}\sum_{m \in \N} - \frac{\princr[\levb]}{1-\princr[\levb]} \\
        & \geq \sum_{\levb \in [\debugg{\frac{\lev - \maxmove}{\di}}, \lev]}\sum_{m \in \N} - \frac{\princr[\levb]}{1-\maxpr} \\
        & = \sum_{\levb \in [\debugg{\frac{\lev - \maxmove}{\di}}, \lev]} - \frac{\sumpr}{\levb(1-\maxpr)} \\
        & \geq - \frac{\sumpr}{1-\maxpr}\sum_{\levb \in [\debugg{\frac{\lev - \maxmove}{\di}}, \lev]} \frac{\di}{\lev-\maxmove} \\
        & = - \frac{\sumpr}{1-\maxpr} \left(\lev- \frac{\lev-\maxmove}{\di}+1\right) \frac{\di}{\lev-\maxmove} \\
        & = -\frac{\sumpr}{1-\maxpr}\frac{(\di-1)\lev+\maxmove+\di}{\lev-\maxmove} \geq - B > 0,
        \end{aligned}
    \end{equation*}
    where $B>0$ is a positive constant which does not depend on $\lev$.
    
    It follows that
    \begin{equation}\label{pr:reinf1}
        \probasup \geq \exp\left(-\debugg{B}\right) > 0.
    \end{equation}
    
    The probability that the upper bound of $\walk$ belongs to $[\lev-\maxmove, \lev]$ is therefore lower bounded by a constant independent of $\lev$ conditionally on the fact that this interval is reached.
    
    \begin{align*}
        \proba[\sup_{\ti} \walk < \infty]& = \sum_{k \geq 1}\proba[\sup_{\ti} \walk \in [(k-1) \maxmove, k \maxmove]]\\ &\geq \sum_{k \geq [\frac{\levz}{M}]+1}\proba[\sup_{\ti}\walk \in[(k-1)\maxmove, k\maxmove]][\reached[k\maxmove]] \proba[\reached[k\maxmove]].
    \end{align*}
    However, if $\walk$ is unbounded, then $\reached[k\maxmove]$ happens, so $\proba[\reached[k\maxmove]] \geq \proba[\sup_\ti \walk = +\infty]$. Therefore, using (\ref{pr:reinf1}):
    \begin{equation*}
        \proba[\sup_{\ti} \walk < \infty] \geq \sum_{k \geq [\frac{\levz}{M}]}  \debugg{\exp\left(-B\right)} \proba[\sup_\ti \walk = +\infty].
    \end{equation*}
    The right hand side is equal to $\infty$ if $\proba[\sup_\ti \walk = +\infty]> 0$, so necessarily $\proba[\sup_\ti \walk = +\infty] = 0$. 
\end{proof}

\end{document}